\documentclass{article}
\usepackage[utf8]{inputenc}
\usepackage{amsmath,amsthm,amssymb}
\usepackage{fullpage}
\usepackage{authblk} 
\usepackage{xcolor}
\usepackage{graphicx}
\usepackage{algorithm}
\usepackage[noend]{algorithmic}

\newtheorem{theorem}{Theorem}

\newtheorem{corollary}[theorem]{Corollary}
\newtheorem{lemma}[theorem]{Lemma}
\newtheorem{myclaim}{Claim}

\theoremstyle{definition}
\newtheorem{definition}[theorem]{Definition}

\title{Obstructions to faster diameter computation: Asteroidal sets
\footnote{This work was supported by a grant of the Romanian Ministry of Research, Innovation and Digitalization, CCCDI - UEFISCDI, project number PN-III-P2-2.1-PED-2021-2142, within PNCDI III.
Results of this paper were partially presented at the IPEC’22 conference~\cite{Duc22}.}}

\author[1,2]{Guillaume Ducoffe}
\affil[1]{\small National Institute for Research and Development in Informatics, Romania}
\affil[2]{\small University of Bucharest, Romania}
\date{}

\begin{document}

\maketitle

\begin{abstract}
{
We propose a novel algorithm for computing the diameter, and approximating all the vertex eccentricities, of an arbitrary graph.
Its runtime outperforms the textbook ${\cal O}(nm)$-time algorithm on $n$-vertex $m$-edge graphs for several superclasses of interval graphs, and even AT-free graphs.
Doing so, we bring interesting new additions to the ``zoo'' of easy graph classes for diameter computation, while we further expand the parameterized complexity of this problem beyond clique-width, tree-width and their relatives.

More specifically, an {\em extremity} is a vertex such that the removal of its closed neighbourhood does not increase the number of connected components.
Let $Ext_{\alpha}$ be the class of all connected graphs whose {\em quotient graph} obtained from modular decomposition contains no more than $\alpha$ pairwise nonadjacent extremities.
The first levels of this increasing hierarchy of graphs already contain well studied classes such as AT-free graphs and {\em dominating pair} graphs of diameter at least six.
Some geometric intersection graphs such as polygon graphs (generalizing the permutation graphs) and chordal graphs of bounded leafage (generalizing the interval graphs) also belong to $Ext_{\alpha}$ for some bounded $\alpha$.
To the best of our knowledge, before our work there was no known deterministic subquadratic-time algorithm for computing the diameter within all these graphs.

Our main contributions are as follows. First, we prove that the diameter of every $m$-edge graph in $Ext_{\alpha}$ can be computed in deterministic ${\cal O}(\alpha^3 m^{3/2})$ time. We then improve our runtime analysis to ${\cal O}(\alpha^3m)$ for bipartite graphs, to ${\cal O}(\alpha^5m)$ for triangle-free graphs, ${\cal O}(\alpha^3\Delta m)$ for graphs with maximum degree $\Delta$, and more generally to linear for all graphs with bounded clique-number. Furthermore, we can compute an additive $+1$-approximation of all vertex eccentricities in deterministic ${\cal O}(\alpha^2 m)$ time.
This is in sharp contrast with general $m$-edge graphs for which, under the Strong Exponential Time Hypothesis (SETH), one cannot compute the diameter in ${\cal O}(m^{2-\epsilon})$ time for any $\epsilon > 0$.

As an important special case of our main result, we derive an ${\cal O}(k^3m^{3/2})$-time algorithm for exact diameter computation within graphs of {\em asteroidal number} at most $k$. Doing so, we considerably extend prior works on exact and approximate diameter computation within AT-free graphs. 
Our techniques are more involved than in these previous works.
Furthermore, our approach is purely combinatorial, that differs from most prior recent works in this area which have relied on geometric primitives such as Voronoi diagrams or range queries.
On our way, we uncover interesting connections between the diameter problem, LexBFS, graph extremities and asteroidal sets. 
We end up presenting an improved algorithm for chordal graphs of bounded asteroidal number, and a partial extension of our results to the class of all graphs with a {\em dominating target} of bounded cardinality.

Our time upper bounds in the paper are shown to be essentially optimal under plausible complexity assumptions.
}
\end{abstract}

\newpage

\section{Introduction}\label{sec:intro}

For any undefined graph terminology, see~\cite{BoM08}.
All graphs considered in this paper are undirected, simple ({\it i.e.}, without loops nor multiple edges) and connected, unless stated otherwise.
Given a graph $G=(V,E)$, let $n = |V|$ and $m = |E|$.
For every vertices $u$ and $v$, let $d_G(u,v)$ be their distance (minimum number of edges on a $uv$-path) in $G$.
Let $e_G(u)$ denote the eccentricity of vertex $u$ (maximum distance to any other vertex).
We sometimes omit the subscript if the graph $G$ is clear from the context.
Finally, let $diam(G) = \max_{u,v \in V} d(u,v) = \max_{u \in V} e(u)$ be the diameter of $G$.

Computing the diameter is important in a variety of network applications such as in order to estimate the maximum latency in communication systems~\cite{DeR94} or to identify the peripheral nodes in some complex networks with a core/periphery structure~\cite{KLPR+05}.
On $n$-vertex $m$-edge graphs, this can be done in ${\cal O}(nm)$ time by running a BFS from every vertex.
This runtime is quadratic in the number $m$ of edges, even for sparse graphs (with $m \leq c \cdot n$ edges for some $c$), and therefore it is too prohibitive for large graphs with millions of vertices and sometimes billions of edges. 
Using Seidel's algorithm~\cite{Sei95}, the diameter of any $n$-vertex graph can be also computed in ${\cal O}(n^{\omega+o(1)})$ time, with $\omega$ the square matrix multiplication exponent.
If $\omega = 2$, then this is almost linear-time for dense graphs, with $m \geq c \cdot n^2$ edges for some constant $c$ (currently, it is only known that $\omega < 2.37286$~\cite{AlV21}).
However, for sparse graphs, this is still in $\Omega(m^2)$.
In~\cite{RoV13}, Roditty and Vassilevska Williams proved that assuming the Strong Exponential-Time Hypothesis of Impagliazzo, Paturi and Zane~\cite{ImP01}, the diameter of $n$-vertex graphs with $n^{1+o(1)}$ edges {\em cannot} be computed in ${\cal O}(n^{2-\epsilon})$ time, for any $\epsilon > 0$.
Therefore, breaking the quadratic barrier for diameter computation is likely to require additional graph structure, even for sparse graphs.
In this paper, we make progress in this direction.

Let us call an algorithm truly subquadratic if it runs in ${\cal O}(m^{2-\epsilon})$ time on $m$-edge graphs, for some fixed positive $\epsilon$. Over the last decades, the existence of a truly subquadratic (often linear-time) algorithm for the diameter problem was proved for many important graph classes~\cite{AVW16,BCD98,BHM20,Cab18,CDV02,CDDH+01,CDP19,Dra93b,DDG21,DrN00,DNB96,Duc21c,Duc19,Duc21b,Duc21d,Duc21a,DHV20,FaP80,GKHM+18,Ola90}.
This has culminated in some interesting connections between faster diameter computation algorithms and Computational Geometry, {\it e.g.}, see the use of Voronoi diagrams for computing the diameter of planar graphs~\cite{Cab18,GKHM+18}, and of data structures for range queries in order to compute all eccentricities within bounded treewidth graphs~\cite{AVW16,BHM20,CaK09}, bounded clique-width graphs~\cite{Duc21a}, or even proper minor-closed graph classes~\cite{DHV20}.
However, this type of geometric approach usually works only if certain Helly-type properties hold for the graph classes considered~\cite{BDS87,BaP99,BoT15,CEV07,Duc21b,DrD21,DHV20}.
Beyond that, the finer-grained complexity of the diameter problem is much less understood, with only a few graph classes for which truly subquadratic algorithms are known~\cite{BCT17}.
The premise of this paper is that the {\em asteroidal number} could help in finding several new positive cases for diameter computation.

\paragraph*{Related work.}
Recall that an independent set in a graph $G$ is a set of pairwise non-adjacent vertices.
An asteroidal set in a graph $G$ is an independent set $A$ with the additional property that, for every vertex $a \in A$, there exists a path between any two remaining vertices of $A \setminus \{a\}$ that does not contain $a$ nor any of its neighbours in $G$.
Let the asteroidal number of $G$ be the largest cardinality of its asteroidal sets. 
The graphs of asteroidal number at most two are sometimes called {\em AT-free graphs}, and they generalize interval graphs, permutation graphs and co-comparability graphs amongst other subclasses~\cite{COS97}. It is worth mentioning here that all the aforementioned subclasses have unbounded treewidth and clique-width.
The properties of AT-free graphs have been thoroughly studied in the literature~\cite{Bei18,BeB01,BFLM15,BKKM99,COS99,COS95,COS97,FMPR04,GHKL+09,GPvL,GoH17,HeK02,KrM12,KMT08,Sta10}, and some of these properties were generalized to the graphs of bounded asteroidal number~\cite{CoS15,KKM97,KKM01}.

In particular, as far as we are concerned here, there is a simple linear-time algorithm for computing a vertex in any AT-free graph whose eccentricity is within one of the diameter~\cite{CDDH+01}.
However, it has been only recently that a truly subquadratic algorithm for {\em exact} diameter computation within this class was presented~\cite{Duc21c}.
This algorithm runs in deterministic ${\cal O}(m^{3/2})$ time on $m$-edge AT-free graphs, and it is combinatorial -- that means, roughly, it does not rely on fast matrix multiplication algorithms or other algebraic techniques.
In fact both algorithms from~\cite{CDDH+01} and~\cite{Duc21c} are based on specific properties of \emph{LexBFS orderings} for the AT-free graphs\footnote{It was also shown in~\cite{CDDH+01} that there exist AT-free graphs such that a multi-sweep LexBFS fails in computing their diameter. Therefore, we need to further process the LexBFS orderings of AT-free graphs, resp. of graphs of bounded asteroidal number, to output their diameter.}.
{
Roughly, the algorithm from~\cite{Duc21c} starts computing a dominating shortest path.
In doing so, the search for a diametral vertex can be restricted to the closed neighbourhood of any one end of this path.
However, in general this neighbourhood might be very large.
The key procedure of the algorithm consists in further pruning out the neighbourhood so that it reduces to a clique.
Then, we are done executing a BFS from every vertex in this clique. 
Both the computation of a dominating shortest path and the pruning procedure of the algorithm are taking advantage of the existence of a linear structure for AT-free graphs, that can be efficiently uncovered by using a double-sweep LexBFS~\cite{COS99}.
Unfortunately, this linear structure no more exists for graphs of asteroidal number $\geq 3$.
}

\paragraph*{Contributions.}
The structure of graphs of bounded asteroidal number, and its relation to LexBFS, is much less understood than for AT-free graphs.
Therefore, extending the known results for the diameter problem on AT-free graphs to the more general case of graphs of bounded asteroidal number is quite challenging.
Doing just that is our main contribution in the paper.
{
In fact, we prove even more strongly that only some types of large asteroidal sets need to be excluded in order to obtain a faster diameter computation algorithm.

More specifically, an {\em extremity} is a vertex such that the removal of its closed neighbourhood leaves the graph connected, see~\cite{KKM01,KrS06}.
Note that every subset of pairwise nonadjacent extremities forms an asteroidal set.
A module in a graph $G=(V,E)$ is a subset of vertices $X$ such that every vertex of $V \setminus X$ is either adjacent to every of $X$ or nonadjacent to every of $X$.
It is a strong module if it does not overlap any other module of $G$.
Finally, the {\em quotient graph} of $G$ is the induced subgraph obtained by keeping one vertex in every inclusionwise maximal strict subset of $V$ which is a strong module of $G$ (see also Sec.~\ref{sec:prelim} for a more detailed discussion about the modular decomposition of a graph).
It is known that except in a few degenerate cases, the diameter of $G$ always equals the diameter of its quotient graph~\cite{CDP19}.
{\em\bf We are interested in the maximum number of pairwise nonadjacent extremities in the quotient graph}, that according to the above is always a lower bound for the asteroidal number. See~\cite[Fig. $1$]{KKM01} for an example where it is smaller than the asteroidal number.

Throughout the paper, let $Ext_{\alpha}$ denote the class of all graphs whose quotient graph contains no more than $\alpha$ pairwise nonadjacent extremities.

\begin{theorem}\label{thm:main}
For every graph $G=(V,E)\in Ext_{\alpha}$, we can compute estimates $\bar{e}(u), \ u \in V$, in deterministic ${\cal O}(\alpha^2m)$ time so that $e(u) \geq \bar{e}(u) \geq e(u) -1$ for every vertex $u$.
In particular, we can compute a vertex whose eccentricity is within one of the diameter.
Moreover, the exact diameter of $G$ can be computed in deterministic ${\cal O}(\alpha^3m^{3/2})$ time.
\end{theorem}

Let us now sketch the main lines of our approach toward proving Theorem~\ref{thm:main}.
First, we replace the input graph by its quotient graph, that can be done in linear time~\cite{TCHP08}.
Then, we compute ${\cal O}(\alpha)$ shortest paths with one common end-vertex $c$, the union of which is a dominating set (to be compared with the dominating shortest path computed in~\cite{Duc21c} for the AT-free graphs).
For that, we prove interesting new relations between graph extremities and LexBFS, but only for graphs that are prime for modular decomposition (this is why we need to consider the quotient graph).
{Roughly, our algorithm computes ${\cal O}(\alpha)$ pairwise nonadjacent extremities, {\it i.e.}, the other end-vertices of the shortest-paths than $c$, by repeatedly executing a modified LexBFS. We stress that our procedure is more complicated than a multi-sweep LexBFS due to the need to avoid getting stuck between two mutually distant extremities.}
In doing so, the search for a diametral vertex can be now restricted to the closed neighbourhoods of only ${\cal O}(\alpha^2)$ vertices (namely, to the ${\cal O}(\alpha)$ furthest vertices from $c$ on every shortest path).
However, unlike what has been done in~\cite{Duc21c} for AT-free graphs, we failed in further pruning out each neighbourhood to a clique.
Instead, we present a new procedure which given a vertex $u$ outputs a vertex in its closed neighbourhood of maximum eccentricity.
This is done by iterating on some extremities at maximum distance from vertex $u$.
Therefore, a key to our analyses in this paper is the number of extremities in a graph.
We provide several bounds on this number.
In doing so, our runtime for exact diameter computation can be improved to ${\cal O}(\alpha^3m)$ time for the bipartite graphs, and more generally to linear time for every graph in $Ext_{\alpha}$ of constant clique number.
We present some more alternative time bounds for our Theorem~\ref{thm:main} in Sec.~\ref{sec:exact}.

\medskip
It is worth noticing that we need {\em not} provide the value $\alpha$ such that the input graph $G$ belongs to $Ext_{\alpha}$.
This is explained in Appendix~\ref{app:compute-alpha}.
In particular, our algorithm is correct for an arbitrary graph $G$.
However, its subquadratic runtime is no more guaranteed.
Our runtime analysis in the paper shows that it mostly depends on two properties, namely: the {\em Gromov hyperbolicity} of the input graph $G$ (see Sec.~\ref{sec:hyp-extrem}), and the number of extremities that are at maximum distance from an arbitrary vertex of $G$.
Many real-life complex networks have a bounded Gromov hyperbolicity~\cite{AbD16}.
It would be interesting to study whether these graphs also have few (pairwise nonadjacent) extremities. 
}

\paragraph*{Matching (Conditional) Lower bounds.}
The algorithm of Theorem~\ref{thm:main} is combinatorial.
In~\cite{Duc21c}, the classic problem of detecting a simplicial vertex within an $n$-vertex graph is reduced in ${\cal O}(n^2)$ time to the diameter problem on ${\cal O}(n)$-vertex AT-free graphs.
The best known combinatorial algorithm for detecting a simplicial vertex in an $n$-vertex graph runs in ${\cal O}(n^3)$ time. 
In the same way, in~\cite{CDDH+01}, Corneil et al. proved an equivalence between the problem of deciding whether an AT-free graph has diameter at most two and a disjoint sets problem which has been recently studied under the name of high-dimensional OV~\cite{DaK21}.
The high-dimensional OV problem can be reduced to Boolean Matrix Multiplication, which is conjectured not to be solvable in ${\cal O}(n^{3-\epsilon})$ time, for any $\epsilon > 0$, using a combinatorial algorithm~\cite{VaW18b}.
It is open whether high-dimensional OV can be solved faster than Boolean Matrix Multiplication.
Therefore, due to both reductions from~\cite{Duc21c} and~\cite{CEV07}, the existence of an ${\cal O}(f(\alpha)m^{3/2-\epsilon})$-time combinatorial algorithm for diameter computation within $Ext_{\alpha}$, for some function $f$ and for some $\epsilon > 0$, would be a significant algorithmic breakthrough. 
 
\paragraph*{Applications to some graph classes.} 
Let us review some interesting subclasses of graphs of $Ext_{\alpha}$, for some constant $\alpha$, for which to the best of our knowledge the best-known deterministic algorithm for diameter computation until this paper has been the brute-force ${\cal O}(nm)$-time algorithm. 

A circle graph is the intersection graph of chords in a cycle. For every $k \geq 2$, a $k$-polygon graph is the intersection graph of chords in a convex $k$-polygon where the ends of each chord lie on two different sides. Note that the $k$-polygon graphs form an increasing hierarchy of all the circle graphs, and that the $2$-polygon graphs are exactly the permutation graphs.
Recently~\cite{Duc21d}, an almost linear-time algorithm was proposed which computes a $+2$-approximation of the diameter of any $k$-polygon graph, for any fixed $k$.
By~\cite{StV18}, every $k$-polygon graph has asteroidal number at most $k$.
Therefore, for the \underline{$k$-polygon graphs}, we obtain an improved $+1$-approximation in linear time, and the first truly subquadratic algorithm for exact diameter computation.

A chordal graph is a graph with no induced cycle of length more than three.
Chordal graphs are exactly the intersection graphs of a collection of subtrees of a host tree~\cite{Gav74}.
We call such a representation a tree model.
The leafage of a chordal graph is the smallest number of leaves amongst its tree models.
In particular, the chordal graphs of leafage at most two are exactly the interval graphs, which are exactly the AT-free chordal graphs.
More generally, every chordal graph of leafage at most $k$ also has asteroidal number at most $k$~\cite{LMW98}.
In~\cite{Duc21c}, a randomized ${\cal O}(km\log^2{n})$-time algorithm was presented in order to compute the diameter of chordal graphs of asteroidal number at most $k$.
Our Theorem~\ref{thm:main} provides a deterministic alternative, but at the price of a higher runtime.
Even more strongly, by combining the ideas of Theorem~\ref{thm:main} with some special properties of chordal graphs, we were able to improve the runtime to ${\cal O}(km)$ -- see our Theorem~\ref{thm:chordal}.
Previously, such as result was only known for interval graphs~\cite{Ola90}.

A {\em moplex} in a graph is a module inducing a clique and whose neighbourhood is a minimal separator (the notions of module and minimal separator are recalled in Sec.~\ref{sec:prelim}). Moplexes are strongly related to LexBFS; indeed, every vertex last visited during a LexBFS is in a moplex~\cite{BeB98}. This has motivated some recent studies on $k$-moplex graphs, {\it a.k.a.}, the graphs with at most $k$ moplexes. In particular, every $k$-moplex graph has asteroidal number at most $k$~\cite{BeB01,DGHK+21+}. Hence, our results in this paper can be applied to the \underline{$k$-moplex graphs}.

{
Finally, a {\em dominating pair} consists of two vertices $x$ and $y$ such that every $xy$-path is a dominating set.
Note that every AT-free graph contains a dominating pair~\cite{COS97}.
A dominating pair graph (for short, DP graph) is one such that every connected induced subgraph contains a dominating pair.
The family $(K_n^+)_{n \geq 4}$ of DP graphs in~\cite[Sec. $4.1$]{PCK04} shows that for every $\alpha \geq 2$, there exists a DP graph which is not in $Ext_{\alpha}$. 
However, we here prove that every DP graph with diameter at least six is in $Ext_2$ -- see our Lemma~\ref{lem:dp}.
In doing so, we obtain a deterministic ${\cal O}(m^{3/2})$-time algorithm which, given an $m$-edge DP graph, either computes its diameter or asserts that its diameter is $\leq 5$.
We left open whether the diameter of DP graphs can be computed in truly subquadratic time.
}

\paragraph*{Organization of the paper.} 
{
We give the necessary graph terminology for this paper in Sec.~\ref{sec:prelim}.
Then, in Sec.~\ref{sec:diam-pair} we present some properties of graph extremities which, to our knowledge, have not been noticed before our work. 
In particular if a graph is prime for modular decomposition, then there always exists a diametral path whose both ends are extremities of the graph. 
We think these results could be helpful in future studies on the diameter problem (for other graph classes), and in order to better understand the relevant graph structure to be considered for fast diameter computation.
We complete Sec.~\ref{sec:diam-pair} with additional properties of extremities for graphs of bounded asteroidal number and for DP graphs.
In Sec.~\ref{sec:framework}, we relate extremities to the properties of Lexicographic Breadth-First Search.
Doing so, we design a general framework in order to compute extremities under various constraints.
We prove Theorem~\ref{thm:main} in Sec.~\ref{sec:alg}, then we discuss some of its extensions in Sec.~\ref{sec:extensions}.
We conclude this paper and propose some open questions in Sec.~\ref{sec:ccl}.

Results of this paper were partially presented at the IPEC'22 conference.
}

\section{Preliminaries}\label{sec:prelim}

We introduce in this section the necessary graph terminology for our proofs.
Let $G=(V,E)$ be a graph.
For any vertex $v \in V$, let $N_G(v) = \{ u \in V \mid uv \in E\}$ be its (open) neighbourhood and let $N_G[v] = N_G(v) \cup \{v\}$ be its closed neighbourhood.
Similarly, for any vertex-subset $S \subseteq V$, let $N_G[S] = \bigcup_{v \in S} N_G[v]$ and let $N_G(S) = N_G[S] \setminus S$.
For any vertices $u$ and $v$, we call a subset $S \subseteq V$ a $uv$-separator if $u$ and $v$ are in separate connected components of $G \setminus S$.
A minimal $uv$-separator is an inclusion-wise minimal $uv$-separator.
We call a subset $S$ a (minimal) separator if it is a (minimal) $uv$-separator for some vertices $u$ and $v$.
Alternatively, a full component for $S$ is a connected component $C$ of $G \setminus S$ such that $N_G(C) = S$.
It is known~\cite{Gol04} that $S$ is a minimal separator if there exist at least two full components for $S$.

\paragraph*{Distances.}
Recall that the distance $d_G(u,v)$ between two vertices $u$ and $v$ equals the minimum number of edges on a $uv$-path.
Let the interval $I_G(u,v) = \{ w \in V \mid d_G(u,v) = d_G(u,w) + d_G(w,v) \}$ contain all vertices on a shortest $uv$-path.
Furthermore, for every $\ell \geq 0$, let $N_G^{\ell}[u] = \{ v \in V \mid d_G(u,v) \leq \ell\}$ be the ball of center $u$ and radius $\ell$ in $G$.
We recall that the eccentricity of a vertex $v \in V$ is defined as $e_G(v) = \max_{u \in V}d_G(u,v)$.
We sometimes omit the subscript if the graph $G$ is clear from the context.
Let $F(v) = \{ u \in V \mid d(u,v) = e(v)\}$ be the set of vertices most distant to vertex $v$.
The diameter and the radius of $G$ are defined as $diam(G) = \max_{v \in V}e(v)$ and $rad(G) = \min_{v \in V}e(v)$, respectively.
We call $(x,y)$ a diametral pair if $d(x,y) = diam(G)$.

\paragraph*{Modular decomposition.}
Two vertices $u,v \in V$ are twins if we have $N(u) \setminus \{v\} = N(v) \setminus \{u\}$.
A twin class is a maximal vertex-subset of pairwise twins.
More generally, a module is a vertex-subset $M \subseteq V$ such that $N(x) \setminus M = N(y) \setminus M$ for any $x,y \in M$.
Note that any twin class is also a module.
We call $G$ {\em prime} if its only modules are: $\emptyset, V \ \text{and} \ \{v\}$ for every $v \in V$ (trivial modules).
A module $M$ is {\em strong} if it does not overlap any other module, {\it i.e.}, for any module $M'$ of $G$, either one of $M$ or $M'$ is contained in the other or $M$ and $M'$ do not intersect. 
We denote by ${\cal M}(G)$ the family of all inclusion wise maximal strong modules of $G$ that do not contain all the vertices of $G$.
Finally, the {\em quotient graph} of $G$ is the graph $G'$ with vertex-set ${\cal M}(G)$ and an edge between every two $M,M' \in {\cal M}(G)$ such that every vertex of $M$ is adjacent to every vertex of $M'$.
The following well-known result is due to Gallai:

\begin{theorem}[\cite{Gal67}]\label{thm:modular-dec}
For an arbitrary graph $G$ exactly one of the following conditions is satisfied.
\begin{enumerate}
\item $G$ is disconnected;
\item its complement $\overline{G}$ is disconnected;
\item or its quotient graph $G'$ is prime for modular decomposition.
\end{enumerate}
\end{theorem}

For general graphs, there is a tree representation of all the modules in a graph, sometimes called the modular decomposition, that can be computed in linear time~\cite{TCHP08}.
Note that since we only consider connected graphs, only the two last items of Theorem~\ref{thm:modular-dec} are relevant to our study.
Moreover, it is easy to prove that if the complement of a graph $G$ is disconnected, then we have $diam(G) \leq 2$.
Therefore, the following result is an easy byproduct of Gallai's theorem for modular decomposition~\cite{Gal67}:

\begin{lemma}[cf. Theorem 14 in~\cite{CDP19}]\label{lem:red-mod}
Computing the diameter (resp., all eccentricities) of any graph $G$ can be reduced in linear time to computing the diameter (resp., all eccentricities) of its quotient graph $G'$.
\end{lemma}

An important observation for what follows is that, if a graph $G$ belongs to $Ext_{\alpha}$ for some $\alpha$, then so does its quotient graph $G'$.
Hence, we may only consider {\em prime} graphs in $Ext_{\alpha}$.

\paragraph*{Hyperbolicity.}

The hyperbolicity of a graph $G$ \cite{Gromov1987} is the smallest half-integer $\delta\ge 0$  such that,  for any four vertices $u,v,w,x$, the two largest of the three distance sums $d(u,v)+d(w,x)$, $d(u,w)+d(v,x)$, $d(u,x)+d(v,w)$ differ by at most $2\delta$. In this case we say that $G$ is $\delta$-hyperbolic. To quote~\cite{DDG21}: ``As the tree-width of a graph measures its combinatorial tree-likeness, so does the hyperbolicity of a graph measure its metric tree-likeness. In other words, the smaller the hyperbolicity $\delta$ of $G$ is, the closer $G$ is to a tree metrically.''

We will use in what follows the following ``tree-likeness'' properties of hyperbolic graphs:

\begin{lemma}[\cite{CDEH+08}]\label{lem:diam-hyp}
If $G$ is $\delta$-hyperbolic, then $diam(G) \geq 2rad(G) - 4\delta -1$.
\end{lemma}

\begin{lemma}[Proposition 3(c) in~\cite{CDHV+19}]\label{lem:diam-furthest}
Let $G$ be a $\delta$-hyperbolic graph and let $u,v$ be a pair of vertices of $G$ such that $v \in F(u)$.
We have $e(v) \geq diam(G) - 8\delta \geq 2rad(G) - 12\delta -1$.
\end{lemma}

\begin{lemma}[\cite{CDEH+08}]\label{lem:rad-hyp}
Let $u$ be an arbitrary vertex of a $\delta$-hyperbolic graph $G$.
If $v \in F(u)$ and $w \in F(v)$, then let $c \in I(v,w)$ be satisfying $d(c,v) = \left\lfloor d(v,w)/2 \right\rfloor$.
We have $e(c) \leq rad(G) + 5\delta$.
\end{lemma}

\section{Properties of graph extremities}\label{sec:diam-pair}

{
We present several simple properties of graph extremities in what follows.
In Sec.~\ref{sec:bounds}, we give bounds on the number of extremities in a graph.
Then, we show in Sec.~\ref{sec:gal} that, for computing the vertex eccentricities of a prime graph (and so, its diameter), it is sufficient to only consider its so-called extremities.
In Sec.~\ref{sec:dom-target}, we relate the location of the extremities in a graph to the one of an arbitrary dominating target.
We state in Sec.~\ref{sec:hyp-extrem} a relationship between extremities and the hyperbolicity of a graph (this result easily follows from~\cite{KKM01}). 
In Sec.~\ref{sec:extrem-graph-classes}, additional properties of extremities in some graph classes are discussed.
}

\subsection{Bounds on the Number of Graph extremities}\label{sec:bounds}

Every non-complete prime graph has at least two extremities~\cite{KKM01}.
The remainder of this section is devoted to proving an upper bound on the number of extremities in a prime graph.
Unfortunately, there may be up to $\Theta(n)$ extremities in an $n$-vertex graph, even if it is AT-free. 
See the construction of~\cite[Fig. $2$]{CDDH+01} for an example. 
It is worth mentioning this example also has clique-number equal to $\Theta(n)$.
Our general upper bounds in what follows show that only dense prime graphs may have $\Omega(n)$ extremities.

\begin{lemma}\label{lem:num-extrem}
If $G \in Ext_{\alpha}$ is prime, then the number of its extremities is at most:
\begin{itemize}
    \item $\alpha \cdot \chi(G)$, where $\chi(G)$ denotes the chromatic number of $G$;
    \item $R(\alpha+1,\omega(G)+1) - 1$, where $\omega(G)$ is the clique number of $G$, and $R(\cdot,\cdot)$ is a Ramsey number.
\end{itemize}
In particular, it is in ${\cal O}(\alpha\sqrt{m})$.
\end{lemma}
\begin{proof}
Let us denote by $q$ the number of extremities of $G$, and let $H$ be induced by all the extremities. Note that $H$ is not necessarily connected. Since we assume that $G \in Ext_{\alpha}$, the independence number of $H$ is at most $\alpha$. 
In this situation, $q < R(\alpha+1,\omega(G)+1)$ (otherwise, either $H$ would contain an independent set of size $\alpha+1$, or $H$ and so, $G$, would contain a clique of size $\omega(G)+1$).
Since the chromatic number of $H$ is at most $\chi(G)$, we also have that $H$ can be partitioned in at most $\chi(G)$ independent sets, and so, $q \leq \alpha \cdot \chi(G)$. In particular, $q = {\cal O}(\alpha\sqrt{m})$ because $\chi(G) = {\cal O}(\sqrt{m})$ for any graph $G$.
\end{proof}

Let $G \in Ext_{\alpha}$ be prime, with $q$ extremities.
Note that, using $R(s,t) = {\cal O}(t^{s-1})$ for any fixed $s$, we get that $q = {\cal O}(\alpha^{\omega(G)+1})$. That is in ${\cal O}(\alpha^3)$ for triangle-free graphs.
For graphs of constant chromatic number, the bound of Lemma~\ref{lem:num-extrem} is linear in $\alpha$.
In particular, $q = {\cal O}(\alpha\Delta)$ for the graphs of maximum degree $\Delta$, and $q = {\cal O}(\alpha)$ for bipartite graphs.

\subsection{Relationships with the diameter}\label{sec:gal}

To the best of our knowledge, the following relation between extremities and vertex eccentricities has not been noticed before:

\begin{lemma}\label{lem:red-extrem}
If $x$ is a vertex of a prime graph $G=(V,E)$ with $|V| \geq 3$, then there exists an extremity $y$ of $G$ such that $d(x,y) = e(x)$.

In particular, for every $y' \in F(x)$, there is an extremity $y \in F(x)$ so that $d(y,y') \leq 2$.
\end{lemma}

We observe that a slightly weaker version of Lemma~\ref{lem:red-extrem} could be also deduced from Lemma~\ref{lem:lexbfs} (proved in the next section).
The following lemma shall be used in our proofs:

\begin{lemma}[\cite{KKM01}]\label{lem:minsep}
Let $S$ be a minimal separator for a prime graph $G=(V,E)$. For any component $C$ of $G \setminus S$, if $C$ does not contain an extremity of $G$, then $N(c) \cap S$ is a separator of $G \setminus C$ for every $c \in C$. 
\end{lemma}

Recall for what follows that a vertex is called universal if and only if all other vertices are adjacent to it.

\begin{proof}[Proof of Lemma~\ref{lem:red-extrem}]
Since we assume that $|V| \geq 3$, we have that $x$ cannot be a universal vertex (otherwise, $V \setminus \{x\}$ would be a nontrivial module, thus contradicting that $G$ is prime).
Let $y \in F(x)$ be arbitrary and we assume that $y$ is not an extremity. 
We shall replace $y$ by some extremity $y^*$ so that $d(x,y) = d(x,y^*) = e(x)$.
First we observe that $x \notin N[y]$ because we assume that $x$ is not a universal vertex.
Let $w$ be disconnected from $x$ in $G \setminus N[y]$, and let $S \subseteq N[y]$ be a minimal $wx$-separator of $G$ (obtained by iteratively removing vertices from $N[y]$ while $w$ and $x$ stay disconnected)
-- possibly, $y \notin S$. --

We continue with a useful property of the connected components of $G \setminus S$.
Specifically, let $C$ be any connected component of $G \setminus S$ not containing vertex $x$.
We claim that $d(x,c) = e(x)$ for each $c \in C$.
Indeed, every shortest $xc$-path contains a vertex $s \in S$ and therefore, $d(x,y) \leq 1 + d(s,x) \leq d(c,s) + d(s,x) = d(c,x)$. In particular, $C \subseteq N(S)$ (otherwise, $d(x,y) < 2 + d(s,x) \leq d(x,c)$ for any $c \in C \setminus N(S)$ and any $s \in S$ that is on a shortest $cx$-path).
It implies that, if $c \in C$ is an extremity, then $d(c,y) \leq 2$.

In what follows, let $X$ the connected component of $x$ in $G \setminus S$.
Then, amongst all vertices of $G \setminus S$ in another connected component than $x$, let $a$ be minimizing $|N(a) \cap S|$ and let $A$ be its connected component in $G \setminus S$. 
Let us assume that $A$ does not contain an extremity of $G$ (else, we are done).
By Lemma~\ref{lem:minsep}, $N(a) \cap S$ is a separator of $G \setminus A$.
Let $b$ be separated from vertex $x$ in $G \setminus (A \cup N(a))$.
We claim that $b \notin S$.
In order to see that, we first need to observe that $X$ is a full component for $S$ (otherwise, $S$ could not be a minimal $wx$-separator).
Therefore if $b \in S$, then the subset $X \cup \{b\}$ would be connected, thus contradicting that $x$ and $b$ are disconnected in $G \setminus (A \cup N(a))$.
This proves our claim, and from now on we denote $B$ the connected component of $b$ in $G \setminus S$.
Observe that $B \neq X$ (otherwise, $N(a) \cap S$ could not be a $bx$-separator in $G \setminus A$).
We further claim that $N(B) \subseteq N(a)$.
Indeed, recall that $X$ is a full component for $S$.
Hence, if it were not the case that $N(B) \subseteq N(a)$ then the subset $X \cup B \cup (N(B) \setminus N(a))$ would be connected, thus contradicting that $x$ and $b$ are disconnected in $G \setminus (A \cup N(a))$.

By the above claim, we get $N(b') \cap S \subseteq N(B) \subseteq N(a) \cap S$, for every $b' \in B$.
Thus, by minimality of $|N(a) \cap S|$, we obtain $N(b') \cap S = N(a) \cap S = N(B)$ for every $b' \in B$.
But then, $B$ is a module of $G$, and therefore $B = \{b\}$ because $G$ is prime.
Finally, suppose by contradiction $b$ is not an extremity.
By repeating the exact same arguments for $b$ instead of $a$, we find another connected component $C=\{c\}$ of $G \setminus S$ so that: $c$ is separated from $x$ in $G \setminus (B \cup N(b)) = G \setminus N[b]$, and $N(c) = N(b) = N(a) \cap S$ (possibly, $C = A$ and $a=c$).
However, it implies that $b$ and $c$ are twins, a contradiction.
Overall, we may choose for our vertex $y^*$ either an extremity of $A$ (if there exists one) or vertex $b$. 
\end{proof}

\begin{corollary}\label{cor"extrem}
If $G=(V,E)$ is prime and $diam(G) \geq 2$, then there exist extremities $x,y$ such that $d(x,y) = diam(G)$.
\end{corollary}

Overall if we were given the $q$ extremities of a prime graph $G$, then by Lemma~\ref{lem:red-extrem}, we could compute all eccentricities (and so, the diameter) in ${\cal O}(qm)$ time. 
By Lemma~\ref{lem:num-extrem}, this runtime is in ${\cal O}(\alpha m^{3/2})$ for the graphs within $Ext_{\alpha}$, which is subquadratic for any fixed $\alpha$.
This bound can be improved to linear time for any graph of $Ext_{\alpha}$ with constant clique number. 
However, the best-known algorithms for computing the extremities run in ${\cal O}(nm)$ time and in ${\cal O}(n^{2.79})$ time~\cite{KrS06}, respectively. 
Furthermore, computing the extremities is at least as hard as triangle detection, even for AT-free graphs~\cite{KrS06}.
We leave as an open problem whether there exists a truly subquadratic algorithm for computing all extremities in a graph.

\subsection{Relationships with Dominating targets}\label{sec:dom-target}

A dominating target in a graph $G$ is a vertex-subset $D$ with the property that any connected subgraph of $G$ containing all of $D$ must be a dominating set.
Dominating targets of cardinality two have been studied under the different name of dominating pairs.
In particular, every AT-free graph contains a dominating pair~\cite{COS97}.

\begin{lemma}\label{lem:dom-target}[special case of Theorems $6$ and $7$ in~\cite{KKM01}]
If $G$ is prime, then every inclusion-wise maximal subset of pairwise nonadjacent extremities in $G$ is a dominating target.
\end{lemma}

We have the following relation between extremities and dominating targets:

\begin{lemma}\label{lem:extrem-close-to-target}
If $D$ is a dominating target of a graph $G$ (not necessarily prime), then every extremity of $G$ is contained in $N[D]$.
In particular, there are at most $(\Delta+1) \cdot |D|$ extremities, where $\Delta$ denotes the maximum degree of $G$.
\end{lemma}
\begin{proof}
Suppose by contradiction the existence of some extremity $v \notin N[D]$. 
Then, $H = G \setminus N[v]$ is a connected subgraph of $G$ that contains all of $D$ but such that $v$ has no neighbour in $V(H)$.
The latter contradicts that $D$ is a dominating target of $G$.
\end{proof}

It follows from both Lemma~\ref{lem:red-extrem} and Lemma~\ref{lem:extrem-close-to-target} that, for any dominating target $D$ in a prime graph, there is a diametral vertex in $N[D]$.
We slightly strengthen this result, as follows.
The following simple lemmas also generalize prior results on AT-free graphs~\cite{CDDH+01} and graphs with a dominating pair~\cite{Duc21c}.

\begin{lemma}\label{lem:ecc-dom-target}
If $D$ is a dominating target, then $F(x) \cap N[D] \neq \emptyset$ for any vertex $x$.

In particular if $F(x) \cap D = \emptyset$, then $F(x) \subseteq N(D)$.
\end{lemma}
\begin{proof}
We may assume that $F(x) \cap D = \emptyset$ (else, we are done).
In this situation, let $y \in F(x)$ be arbitrary.
Then, let $H$ be the union of shortest $xu$-paths, for every $u \in D$.
Since $H$ is a connected subgraph, we have $y \in N[H]$.
In particular, there is a shortest $ux$-path $P$, for some fixed $u \in D$, such that $y \in N[P]$.
Observe that $y \notin V(P)$ (otherwise, $d(x,y) \leq d(x,u)$, and therefore $u \in F(x)$).
So, let $y^* \in V(P) \cap N(y)$.
If $y^* \neq u$, then $d(x,y) \leq d(x,y^*) + 1 \leq (d(x,u)-1)+1 = d(x,u)$, and therefore, $u \in F(x)$. A contradiction.
As a result, $y \in N(u) \setminus D \subseteq N(D)$.
\end{proof}

\begin{corollary}\label{cor:neighbours-dom-target}
If $D$ is a dominating target of a graph $G$, and no vertex of $D$ is in a diametral pair, then $x,y \in N(D)$ for every diametral pair $(x,y)$.
\end{corollary}
\begin{proof}
Let $(x,y)$ be an arbitrary diametral pair of $G$.
In order to prove the result, by symmetry, it suffices to prove that $y \in N(D)$.
Since $y \in F(x)$ and $D \cap F(x) = \emptyset$ (else, some vertex of $D$ would belong to a diametral pair), it follows from Lemma~\ref{lem:ecc-dom-target}.
\end{proof}

This above Corollary~\ref{cor:neighbours-dom-target} suggests the following strategy in order to compute the diameter of a prime graph $G$.
First, we compute a small dominating target $D$.
Then, we search for a diametral vertex within the neighbourhood of each of its $|D|$ vertices.
If $G \in Ext_{\alpha}$, then according to Lemma~\ref{lem:dom-target} there always exists such a $D$ with ${\cal O}(\alpha)$ vertices.
However, we are not aware of any truly subquadratic algorithm for computing this dominating target.
By Lemma~\ref{lem:dom-target}, it is sufficient to compute a maximal independent set of extremities, but then we circle back to the aforementioned problem of computing all extremities in a graph.
In Sec.~\ref{sec:alg} in the paper, we prove that we needn't compute a dominating target in full in order to determine what the diameter is.
Specifically, we may only compute a strict subset $D' \subset D$ of a dominating target (for that, we use the techniques presented in Sec.~\ref{sec:framework}).
However, the price to pay is that while doing so, we also need to consider a bounded number of vertices outside of $N[D']$ and their respective neighbourhoods.
Hence, the number of neighbourhoods to be considered grows to ${\cal O}(\alpha^2)$.
This will be our starting approach for proving Theorem~\ref{thm:main}.

Finally, we want to stress here that the property of having a small dominating target is not sufficient on its own for ensuring faster diameter computation algorithms.
Indeed, assuming the Strong Exponential-Time Hypothesis, the diameter of $n$-vertex graphs with $n^{1+o(1)}$ edges and a dominating edge cannot be computed in ${\cal O}(n^{2-\epsilon})$ time, for any $\epsilon > 0$~\cite{Duc21c}. For partial extensions of our results to the graphs with a bounded-cardinality dominating target, see Sec.~\ref{sec:extensions}.

\subsection{Relationships with Hyperbolicity}\label{sec:hyp-extrem}

{
Recall the definition of $\delta$-hyperbolic graphs in Sec.~\ref{sec:prelim}.
In a $\delta$-hyperbolic graph $G$, an ``almost central'' vertex of eccentricity $\leq rad(G) + c \delta$, for some $c > 0$, can be computed in linear time, using a double-sweep BFS (see Lemma~\ref{lem:rad-hyp}). Then, according to Lemma~\ref{lem:diam-hyp}, any diametral vertex must at a distance $\geq rad(G) - c' \delta$, for some $c' > 0$, to this almost central vertex.
Roughly, we wish to combine these properties with the computation of some subset $D'$ of a small dominating target (see Sec.~\ref{sec:dom-target}) in order to properly locate some neighbourhood that contains a diametral vertex.
For that, we need to prove here that graphs in $Ext_{\alpha}$ are $\delta$-hyperbolic for some $\delta$ depending on $\alpha$.
Namely:

\begin{lemma}\label{lem:hyp-ext-alpha}
Every graph $G \in Ext_{\alpha}$ is $(3\alpha-1)$-hyperbolic.
\end{lemma}
\begin{proof}
This result directly follows from the combination of several prior works.
First, if the quotient graph $G'$ of $G$ is $\delta$-hyperbolic, for some $\delta \geq 1$, then so is $G$~\cite{Sot11}.
Thus, from now on we assume (up to replacing $G$ by $G'$) that graph $G$ is prime.
By Lemma~\ref{lem:dom-target}, $G$ contains a dominating target $D$ with $|D| \leq \alpha$ vertices.
Furthermore, if $D$ denotes a dominating target in a graph $G$, then $G$ admits an additive tree $(3|D|-1)$-spanner~\cite{KKM01}.
Since graphs with a tree $t$-spanner are $t$-hyperbolic~\cite{CDEH+08,DrK14}, we obtain as desired that $G$ is $(3\alpha-1)$-hyperbolic.
\end{proof}

We can improve this upper bound on the hyperbolicity for graphs of bounded asteroidal number.
Specifically, a graph is $k'$-chordal if it has no induced cycle of length more than $k'$. 
Every graph of asteroidal number at most $k$ must be $(2k+1)$-chordal. 
Furthermore, for any $k' \geq 4$, every $k'$-chordal graph is $\frac 1 2 \left\lfloor k'/2 \right\rfloor$-hyperbolic~\cite{WuZ11}.
Therefore, we obtain that the graphs of asteroidal number at most $k$ are $k/2$-hyperbolic.
}

\subsection{Extremities in some Graph classes}\label{sec:extrem-graph-classes}

We complete this section with the following inclusions between graph classes.

\begin{lemma}\label{lem:an-k}
Every graph $G$ of asteroidal number $k$ belongs to $Ext_k$.
\end{lemma}

While this above Lemma~\ref{lem:an-k} trivially follows from the respective definitions of $Ext_k$ and the asteroidal number, the following result is less immediate: 
 
\begin{lemma}\label{lem:dp}
Every DP graph $G=(V,E)$ of diameter at least six belongs to $Ext_2$.
\end{lemma}
\begin{proof}
Since the property of being a DP graph is hereditary, the quotient graph of any DP graph is also a DP graph.
In particular, it suffices to prove even more strongly that an arbitrary DP graph $G$ of diameter at least six (not necessarily prime) cannot contain three pairwise nonadjacent extremities.
Suppose by contradiction the existence of three such extremities $u,v,w$.
Let $(x,y)$ be a dominating pair.
By Lemma~\ref{lem:extrem-close-to-target}, $u,v,w \in N[x] \cup N[y]$.
Without loss of generality, let $u,v \in N(x)$.

We claim that $S=N(u) \cap N(v)$ is not a separator of $G$.
Suppose by contradiction that it is the case.
Let $A = N[u] \setminus N(v), \ B=N[v] \setminus N(u) \ \text{and} \ X = V \setminus \left(A \cup B \cup S\right)$.
Since $S \subseteq N(u)$ and $u$ is an extremity of $G$, $B \cup X$ must be contained in some connected component of $G \setminus S$.
But similarly, since $S \subseteq N(v)$ and $v$ is an extremity of $G$, $A \cup X$ must be also contained in some connected component of $G \setminus S$.
As a result, $X = \emptyset$, and the only two components of $G \setminus S$ are $A$ and $B$.
In particular, $w \in A \cup B \subseteq N[u] \cup N[v]$, that is a contradiction.
Therefore, we proved as claimed that $S$ is not a separator of $G$.

We now claim that $u,v$ are still extremities in the subgraph $G \setminus S$.
By symmetry, it suffices to prove the result for vertex $u$.
Observing that removing all of $N[u] \setminus S$ leaves us with $G \setminus N[u]$, we are done because $u$ is an extremity of $G$. -- However, please note that vertex $w$ may not be an extremity of $G \setminus S$. --

Since $G \setminus S$ is a DP graph, there exists a dominating pair $(x',y')$ in this subgraph.
Again by Lemma~\ref{lem:extrem-close-to-target} we have $u,v \in N[x'] \cup N[y']$.
Without loss of generality, let $u \in N[x'], \ v \in N[y']$ (possibly, $u = x'$, resp. $v = y'$).
Now, let $P=(z_0=x',z_1,\ldots,z_{\ell}=y')$ be a shortest $x'y'$-path of $G \setminus S$.
By construction, $P$ is a dominating path of $G \setminus S$.
In order to derive a contradiction, we shall prove, using $P$, that $e(x) \leq 4$. 
Indeed, doing so, since $(x,y)$ is a dominating pair, we obtain that $diam(G) \leq e(x) + 1 \leq 5$, a contradiction.
For that, let $t \in V$ be arbitrary. 
We may further assume $t \notin N[x]$.
If $t \in S$, then $u,v \in N(t) \cap N(x)$, therefore $d(x,t) \leq 2$.
From now on, let us assume that $t \notin S$. Consider some index $i$ such that $t \in N[z_i]$.
Let $Q_u$ be an induced $xz_i$-path such that $V(Q_u) \subseteq \{x,u,z_0,z_1,\ldots,z_i\}$.
In the same way, let $Q_v$ be an induced $xz_i$-path such that $V(Q_v) \subseteq \{x,v,z_{\ell},z_{\ell-1},\ldots,z_i\}$.
Let us first assume that $V(Q_u) \cup V(Q_v)$ induces a cycle $C$.
Then, the length of $C$ must be $\leq 6$ because $C$ is a DP graph and no cycle of length $\geq 7$ contains a dominating pair.
As a result, $d(x,z_i) \leq 3$, and so $d(x,t) \leq 4$. 
For the remainder of the proof, we assume that there exists a chord in the cycle $C$ induced by $V(Q_u) \cup V(Q_v)$.
Since the three of $P,Q_u,Q_v$ are induced paths, the only possible chords are: $uz_j$, for some $i+1 \leq j \leq \ell$; or $vz_j$, for some $0 \leq j \leq i-1$.
By symmetry, let $uz_j$ be a chord of $C$.
Since $P$ is a shortest $x'y'$-path of $G \setminus S$, and $u \in N[x']$, we obtain that $j \in \{0,1,2\}$.
In particular, we obtain that $i \leq 1$, and so, $d(x,t) \leq 1 + d(x,z_i) \leq 2+d(u,z_i) \leq 4$.
\end{proof}

Lemma~\ref{lem:dp} does not hold for diameter-five DP graphs, as it can be shown from the example in~\cite[Fig. 6]{PCK04}, that has three pairwise nonadjacent extremities.
Moreover, for every $n \geq 4$, there exists a diameter-two DP graph $K_n^+$ with $n$ pairwise nonadjacent extremities~\cite{PCK04}.
We left open whether, for any $d \in \{3,4,5\}$, there exists some constant $\alpha(d) \geq 3$ such that all diameter-$d$ DP graphs belong to $Ext_{\alpha(d)}$.

\section{A framework for computing extremities}\label{sec:framework}

We identify sufficient conditions for computing an independent set of extremities (not necessarily a maximal one).
To the best of our knowledge, before this work there was no faster known algorithm for computing {\em one} extremity than for computing all such vertices.
We present a simple linear-time algorithm for this problem on prime graphs --- see Sec.~\ref{sec:one-extrem}.
Then, we refine our strategy in Sec.~\ref{sec:gal-extrem} so as to compute one extremity avoided by some fixed connected subset.
This procedure is key to our proof of Theorem~\ref{thm:main}, for which we need to iteratively compute extremities, and connect those to some pre-defined vertex $c$ using shortest paths, until we obtain a connected dominating set.
In fact, our approach in Sec.~\ref{sec:gal-extrem} works under more general conditions which we properly state in Def.~\ref{def:transitive}.
Our main algorithmic tool here is LexBFS, of which we first recall basic properties in Sec.~\ref{sec:lexbfs}.

\subsection{LexBFS}\label{sec:lexbfs}
The Lexicographic Breadth-First Search (LexBFS) is a standard algorithmic procedure, that runs in linear time~\cite{RTL76}. We give a pseudo-code in Algorithm~\ref{alg:lexbfs}. Note that we can always enforce a start vertex $u$ by assigning to it an initial non empty label. Then, for a given graph $G = (V,E)$ and a start vertex $u$, $LexBFS(u)$ denotes the corresponding execution of LexBFS. Its output is a numbering $\sigma$ over the vertex-set (namely, the reverse of the ordering in which vertices are visited during the search). In particular, if $\sigma(i) = x$, then $\sigma^{-1}(x) = i$. 

%
\begin{algorithm}
	\caption{LexBFS~\cite{RTL76}.}
	\label{alg:lexbfs}
	\footnotesize
	
	\begin{algorithmic}[1]
		\REQUIRE{A graph $G=(V,E)$.}\\
		    \STATE{assign the label $\emptyset$ to each vertex;}
		    \FOR{$i=n$ {\bf to} $1$}
		        \STATE{pick an unnumbered vertex $x$ with the largest label in the lexicographic order;}
		        \FORALL{unnumbered neighbours $y$ of $x$}
		            \STATE{add $i$ to $label(y)$;}
		        \ENDFOR
		        \STATE{$\sigma(i) \leftarrow x$ {\it /* number $x$ by $i$ */};}
		    \ENDFOR
	\end{algorithmic}
\end{algorithm}

We use some notations from~\cite{COS99}.
Fix some LexBFS ordering $\sigma$.
Then, for any vertices $u$ and $v$, $u \prec v$ if and only if $\sigma^{-1}(u) < \sigma^{-1}(v)$.
Similarly, $u \preceq v$ if either $u = v$ or $u \prec v$.
Let us define $N_{\prec}(v) = \{ u \in N(v) \mid u \prec v \}$ and $N_{\succ}(v) = \{ u \in N(v) \mid v \prec u \}$.
Let also $\lhd$ denote the lexicographic total order over the sets of LexBFS labels.
For every vertices $u$ and $v$, $u \preceq v$, let $\lambda(u,v)$ be the label of vertex $u$ when vertex $v$ was about to be numbered. We stress that $\lambda(u,v) \trianglelefteq \lambda(v,v)$ ({\it i.e.}, the vertex selected to be numbered at any step has maximum label for the lexicographic order).
Furthermore, a useful observation is that $\lambda(u,v)$ is just the list of all neighbours of $u$ which got numbered before $v$, ordered by decreasing LexBFS number.
In particular, for any $u \preceq v$ we have $\lambda(u,v) = \lambda(v,v)$ if and only if $N_{\succ}(v) \subseteq N(u)$.
We often use this latter property in our proofs.

\begin{lemma}[monotonicity property~\cite{COS99}]\label{lem:monotone}
Let $a,b,c$ and $d$ be vertices of a graph $G$ such that: $a \preceq c$, $b \preceq c$ and $c \prec d$.
If $\lambda(a,d) \lhd \lambda(b,d)$, then $\lambda(a,c) \lhd \lambda(b,c)$.
\end{lemma}

\begin{corollary}\label{cor:monotone}
Let $x,y,z$ be vertices of a graph $G$ such that: $x \preceq y \preceq z$, and $\lambda(x,z) = \lambda(z,z)$.
Then, $\lambda(y,z) = \lambda(z,z)$.
\end{corollary}
\begin{proof}
Suppose by contradiction $\lambda(y,z) \neq \lambda(z,z)$.
In particular, $\lambda(y,z) \lhd \lambda(z,z)$, that is equivalent to $\lambda(y,z) \lhd \lambda(x,z)$.
By Lemma~\ref{lem:monotone} applied to $(a,b,c,d) = (y,x,y,z)$, we get $\lambda(y,y) \lhd \lambda(x,y)$, a contradiction.
\end{proof}

\subsection{Finding one extremity}\label{sec:one-extrem}

It turns out that finding {\em one} extremity is simple, namely:

\begin{lemma}\label{lem:lexbfs}
If $G=(V,E)$ is a prime graph with $|V| \geq 3$, and $\sigma$ is any LexBFS order, then $v = \sigma(1)$ is an extremity of $G$.
\end{lemma}
\begin{proof}
%
Suppose by contradiction $G \setminus N[v]$ to be disconnected.
Let  $u = \sigma(n)$ be the start vertex of the LexBFS ordering, and let $C$ be any component of $G \setminus N[v]$ which does not contain vertex $u$.
We denote by $z = \sigma(i), \ n > i > 1$ the vertex of $C$ with maximum LexBFS number.
By maximality of $z$, we obtain that $N_{\succ}(z) \subseteq N(v)$.
In particular, $\lambda(v,z) = \lambda(z,z)$.
Then, let $M = \{ w \in V \mid v \preceq w \preceq z \}$.
Since we have $\lambda(v,z) = \lambda(z,z)$, by Corollary~\ref{cor:monotone} we also get $\lambda(w,z) = \lambda(z,z)$ for every $w \in M$.
But this implies $N(w) \setminus M = N_{\succ}(z)$ for each $w \in M$, therefore $M$ is a nontrivial module of $G$. 
A contradiction. 
\end{proof}

An alternative proof of Lemma~\ref{lem:lexbfs} could be deduced from the work of Berry and Bordat on the relations between moplexes and LexBFS~\cite{BeB00}.
However, to the best of our knowledge, Lemma~\ref{lem:lexbfs} has not been proved before.

\subsection{Generalization}\label{sec:gal-extrem}

Before generalizing Lemma~\ref{lem:lexbfs}, we need to introduce a few more notions and terminology.

\begin{definition}\label{def:dominance}
Let $G=(V,E)$ be a graph and let $u,v,w \in V$ be pairwise independent. 
We write $u \perp_w v$ if and only if $u,v$ are in separate connected components of $G \setminus N[w]$.
\end{definition}

A vertex $w$ intercepts a path $P$ if $N[w] \cap V(P) \neq \emptyset$, and it misses $P$ otherwise.
We can easily check that if $u \perp_w v$, then $w$ intercepts all $uv$-paths, and conversely if $u \not\perp_w v$ and $uw, vw \notin E$, then $w$ misses a $uv$-path.
Furthermore, we stress that for every fixed vertex $u$ and every vertex $w \notin N[u]$, $w$ is an extremity of $G$ if and only if there is no vertex $v$ such that $u \perp_w v$.

\begin{definition}\label{def:transitive}
Let $u$ and $S$ be, respectively, a vertex and a vertex-subset of some graph $G=(V,E)$.
We call $S$ a $u$-transitive set if, for every $x \in S$ and $y \in V$ nonadjacent, $x \perp_y u \Longrightarrow y \in S$.
\end{definition}

Next, we give examples of $u$-transitive sets.

\begin{lemma}\label{lem:transitive}
If $H$ is a connected subgraph of a graph $G$, then its closed neighbourhood $N[H]$ is $u$-transitive for every $u \in V(H)$.
In particular, every ball centered at $u$ and of arbitrary radius is $u$-transitive.
\end{lemma}
\begin{proof}
Let $x \in N[H]$ and $y$ satisfy $x \perp_y u$.
Since $x \in N[H]$ and $u \in V(H)$, there exists a $xu$-path $P$ of which all vertices except maybe $x$ are in $H$.
Furthermore, since we assume that $x \perp_y u$, and so $x$ and $y$ are nonadjacent, we obtain $y \in N[P \setminus \{x\}] \subseteq N[H]$.
\end{proof}

For an example of non-connected $u$-transitive set, we may simply consider three pairwise nonadjacent vertices $u,v,w$ in a cycle. Then, $S=\{v,w\}$ is $u$-transitive.

\medskip
We are now ready to state the following key lemma:

\begin{lemma}\label{lem:next-extrem}
Let $G = (V,E)$ be a prime graph with $|V| \geq 3$, let $u \in V$ be arbitrary and let $S \subseteq V$ be $u$-transitive.
If $V \neq S \cup N[u]$, then we can compute in linear time an extremity $v \notin S \cup N[u]$ such that $d(u,v)$ is maximized.
\end{lemma}
\begin{proof}
We describe the algorithm before proving its correctness and analysing its runtime.

\smallskip
{\bf Algorithm.}
Let $\sigma$ be any LexBFS($u$) order of $G$, and let $w \notin S$ be minimizing $\sigma^{-1}(w)$.
First we compute the maximum index $i$, $n > i \geq \sigma^{-1}(w)$, so that $\lambda(w,\sigma(i)) = \lambda(\sigma(i),\sigma(i)) = \lambda_i$.
We set $j := 0$, $S_j := S$ and $M_j := \{ v \notin S \mid w \preceq v \preceq \sigma(i) \}$.
Then, while $|M_j| > 1$, we apply the following procedure:
\begin{itemize}
    \item We partition $M_j$ into groups $A_j^1,A_j^2,\ldots,A_j^{p_j}$ so that two vertices are in the same group if and only if they have the same neighbours in $S_j$. 
    \item Without loss of generality let $A_j^1$ be minimizing $|N(A_j^1) \cap S_j|$. 
    We set $M_{j+1} = A_j^1$, $S_{j+1} = M_j \setminus M_{j+1}$.
    \item We set $j := j+1$.
\end{itemize}
If $|M_j| = 1$ (end of the while loop), then we output the unique vertex $v \in M_j$.

\medskip
{\bf Correctness.}
We prove by induction that the following three properties hold, for any $j \geq 0$: {\tt(i)} $M_j$ is a module of $G \setminus S_j$; {\tt(ii)} every $v \in M_j$ is a vertex of $V \setminus S$ such that $d(u,v)$ is maximized; and {\tt(iii)} for every $v \in M_j$, either $v$ is an extremity of $G$ or every connected component $C$ of $G \setminus N[v]$ that does not contain vertex $u$ satisfies $C \subseteq M_j$.

First, we consider the base case $j=0$.
Recall that $M_0  = \{ v \notin S \mid w \preceq v \preceq \sigma(i)  \}$, where $i$ is the maximum index such that $\lambda(w,\sigma(i)) = \lambda(\sigma(i),\sigma(i)) = \lambda_i$.
By Corollary~\ref{cor:monotone} we have $\lambda(v,\sigma(i)) = \lambda_i$ for each $v \in M_0$, and so, $N_{\succ}(\sigma(i)) \subseteq N(v)$.
The latter implies that $d(u,v) = d(u,\sigma(i))$ because $\sigma$ is a (Lex)BFS order.
Equivalently, $d(u,v) = d(u,w)$, and by the minimality of $\sigma^{-1}(w)$ we have that $w$ is a vertex of $V \setminus S$ maximizing $d(u,w)$ (Property {\tt(ii)}).
Moreover, $N(v) \setminus (S \cup M_0) = N_{\succ}(\sigma(i)) \setminus S$ for each $v \in M_0$, therefore $M_0$ is a module of $G \setminus S$ (Property {\tt(i)}).
Now, let $v \in M_0$ be arbitrary.
If $v$ is not an extremity of $G$, then let $C$ be any connected component of $G \setminus N[v]$ not containing vertex $u$. We claim that $C \subseteq M_0$. 
Indeed, suppose by contradiction $C \cap S \neq \emptyset$.
By maximality of $d(u,v)$, we have $v \notin N[u]$ (for else, $V \setminus S \subseteq N[u]$).
But then, we would get $s \perp_v u$ for any $s \in C \cap S$, and therefore by the definition of $S$ we should have $v \in S$.
A contradiction.
As a result we have $C \cap S = \emptyset$.
Furthermore, we have $w \preceq c$ for each $c \in C$, that follows from the minimality of $\sigma^{-1}(w)$.
Let $z \in C$ be maximizing $\sigma^{-1}(z)$.
In order to prove that $C \subseteq M_0$, it now suffices to prove that $\sigma^{-1}(z) \leq i$.
Suppose by contradiction that it is not the case.
We first observe $N_{\succ}(z)  \subseteq N(v)$ by maximality of $\sigma^{-1}(z)$. 
Therefore (since in addition, $v \preceq \sigma(i) \prec z$), $\lambda(v,z) = \lambda(z,z)$.
Since we suppose $v \preceq \sigma(i) \prec z$, we also get by Corollary~\ref{cor:monotone} that $\lambda(\sigma(i),z) = \lambda(z,z)$.
Then, $N_{\succ}(z) \subseteq N_{\succ}(\sigma(i)) \subseteq N(w)$.
However, it implies $\lambda(w,z) = \lambda(z,z)$, thus contradicting the maximality of $i$ ($< \sigma^{-1}(z)$) for this property.

Then, let us assume that $M_j,S_j$ satisfy all of properties {\tt(i)}, {\tt(ii)} and {\tt(iii)}, and that $|M_j| > 1$.
By construction, all vertices in $M_{j+1}$ have the same neighbours in $S_j$.
Since $M_{j+1} \subset M_j$ and $M_j$ is a module of $G \setminus S_j$, we obtain that $M_{j+1}$ is a module of $G \setminus \left(M_j \setminus M_{j+1}\right) = G \setminus S_{j+1}$ (Property {\tt(i)}).
Property {\tt(ii)} also holds because it holds for $M_j$ and $M_{j+1} \subset M_j$.
Now, let $v \in M_{j+1}$ be arbitrary, and let us assume it is not an extremity of $G$.
Let $C$ be any component of $G \setminus N[v]$ not containing vertex $u$.
By Property {\tt(iii)}, $C \subseteq M_j$.
Suppose by contradiction $C \not\subset M_{j+1}$.
Let $z \in C \setminus M_{j+1}$. Since we have $C \cap S_j = \emptyset$, we obtain $N(z) \cap S_j \subseteq N(v) \cap S_j$.
By minimality of $|N(v) \cap S_j| = |N(M_{j+1}) \cap S_j|$, we get $N(z) \cap S_j = N(v) \cap S_j$, which contradicts that $z \notin M_{j+1}$.

\smallskip
The above Property {\tt(iii)} implies that, if $|M_j| = 1$, then the unique vertex $v \in M_j$ is indeed an extremity.
Furthermore, by Property {\tt(ii)}, $v$ is a vertex of $V \setminus S$ that maximizes $d(u,v)$.
Hence, in order to prove correctness of the algorithm, all that remains to prove is that this algorithm eventually halts.
For that, we claim that if $|M_j| > 1$ then $|M_{j+1}| < |M_j|$.
Indeed, Property {\tt (i)} asserts that $M_j$ is a module of $G \setminus S_j$.
Let $A_j^1,A_j^2,\ldots,A_j^{p_j}$ be the partition of $M_j$ such that two vertices are in the same group if and only if they have the same neighbours in $S_j$.
Since $G$ is prime, $M_j$ cannot be a nontrivial module of $G$, and therefore $p_j \geq 2$.
Hence, $|M_{j+1}| = |A_j^1| < |M_j|$, as claimed.
This above claim implies that eventually we reach the case when $|M_j| = 1$, and so, the algorithm eventually halts.

\smallskip
{\bf Complexity.}
Computing the LexBFS ordering $\sigma$ can be done in linear time~\cite{RTL76}.
Then once we computed vertex $w$ in additional ${\cal O}(n)$ time, we can compute the largest index $i$ such that $\lambda(w,\sigma(i)) = \lambda(\sigma(i),\sigma(i))$ as follows.
We mark all the neighbours of vertex $w$, then we scan the vertices by decreasing LexBFS number, and we stop at the first encountered vertex $x \neq u$ such that all vertices in $N_{\succ}(x)$ are marked.
Since all the neighbour-sets need to be scanned at most once, the total runtime for this step is linear.
Finally, we dynamically maintain some partition such that, at the beginning of any step $j \geq 0$, this partition equals $(M_j)$.
If $|M_j| > 1$, then we consider each vertex $s \in S_j$ sequentially, and we replace every group $X$ in the partition by the nonempty groups amongst $X \setminus N(s),X \cap N(s)$. 
In doing so, we obtain the partition $(A_j^1,A_j^2,\ldots,A_j^{p_j})$.
We remove all groups but $M_{j+1}$, computing $S_{j+1} = M_j \setminus M_{j+1}$ along the way.
By using standard partition refinement techniques~\cite{HMPV00,PaT87}, after an initial processing in ${\cal O}(|M_0|)$ time each step $j$ can be done in ${\cal O}\left(\sum_{s \in S_j}|N(s)| + |M_j \setminus M_{j+1}|\right)$ time. 
Because all the sets $S_j$ are pairwise disjoint the total runtime is linear.
\end{proof}

\section{Proof of Theorem~\ref{thm:main}}\label{sec:alg}

In Sec.~\ref{sec:approx}, we present a linear-time algorithm for computing a vertex whose eccentricity is within one of the true diameter.
This part of the proof is simpler, and it gives some intuition for our exact diameter computation algorithm, which we next present in Sec.~\ref{sec:exact}.

For simplicity, the following Theorems~\ref{thm:approx} and~\ref{thm:exact} assume that each input graph $G$ is given with some value $\alpha$ so that $G \in Ext_{\alpha}$.
However, this stringent assumption can be removed up to slight modifications of our main algorithms.
We postpone the details to Appendix~\ref{app:compute-alpha}.

\subsection{Approximation algorithm}\label{sec:approx}

\begin{theorem}\label{thm:approx}
For every graph $G=(V,E) \in Ext_{\alpha}$, we can compute in deterministic ${\cal O}(\alpha^2m)$ time estimates $e(v) \geq \bar{e}(v) \geq e(v) -1$ for every vertex $v$.
\end{theorem}
\begin{proof}
We may assume $G$ to be prime by Lemma~\ref{lem:red-mod}.
Furthermore, let us assume that $|V| \geq 3$.
We subdivide the algorithm in three main phases.
\begin{itemize}
    \item First, we compute some shortest path by using a double-sweep LexBFS.
    More specifically, let $x_1$ be the last vertex numbered in a LexBFS.
    Let $x_2$ be the last vertex numbered in a LexBFS($x_1$).
    We compute a vertex $c \in I(x_1,x_2)$ so that $d(c,x_1) = \left\lfloor d(x_1,x_2)/2\right\rfloor$.
    Then, let $P_1$ (resp., $P_2$) be an arbitrary shortest $cx_1$-path (resp., $cx_2$-path).
    -- Note that for AT-free graphs (but not necessarily in our case), the shortest $x_1x_2$-path $P_1 \cup P_2$ is dominating~\cite{COS99}. --
    \item Second, we set $H := P_1 \cup P_2$. While $H$ is not a dominating set of $G$, we compute an extremity $x_i \notin N[H]$ and we add an arbitrary shortest $cx_i$-path to $H$. -- We stress that such an extremity $x_i$ always exists due to $H$ being connected and therefore $N[H]$ being $c$-transitive (see Lemma~\ref{lem:transitive}) and by Lemma~\ref{lem:next-extrem}.-- Let $x_3,x_4,\ldots,x_t$ denote all the extremities computed. By construction, $H$ is the union of $t$ shortest paths $P_1,P_2,\ldots,P_t$ with one common end-vertex $c$.
    \item Finally, for every $1 \leq i \leq t$, let the subset $U_i$ be composed of the $\min\{d(x_{i},c)+1,66\alpha-19\}$ closest vertices to $x_i$ in $P_i$ (including $x_i$ itself).
    Let $U = \bigcup_{i=1}^t U_i$.
    For every vertex $v \in V$, we set $\bar{e}(v) := \max\{d(u,v) \mid u \in U\}$.
\end{itemize}

\smallskip
\noindent
{\bf Correctness.}
Let $v \in V$ be arbitrary.
Since $H$ is a dominating set of $G$, some vertex $u \in H$, in the closed neighbourhood of any vertex of $F(v)$, must satisfy $d(u,v) \geq e(v) -1$.
In order to prove correctness of our algorithm, it suffices to prove the existence of one such vertex in $U$.
For that, let $\delta$ be chosen such that $G$ is $\delta$-hyperbolic.
By Lemma~\ref{lem:rad-hyp}, $e(c) \leq  rad(G) + 5\delta$.
Furthermore if $u \in H$ satisfies $N[u] \cap F(v) \neq \emptyset$, then by Lemma~\ref{lem:diam-furthest}, $e(u) \geq 2rad(G) - 12\delta -2$.
In particular, $d(u,c) \geq rad(G) - 17\delta - 2$ (else, $e(u) \leq d(u,c) + e(c) \leq 2rad(G) - 12\delta - 3$).
For $1 \leq i \leq t$ such that $u \in V(P_i)$, since $P_i$ is a shortest $x_ic$-path of length at most $e(c) \leq rad(G) + 5\delta$, we get that $u$ must be one of the $(rad(G)+5\delta)-(rad(G)-17\delta-2)+1=22\delta+3$ closest vertices to $x_i$.
By Lemma~\ref{lem:hyp-ext-alpha}, $\delta \leq 3\alpha-1$, therefore $22\delta+3 \leq 66\alpha-19$.

\smallskip
\noindent
{\bf Complexity.} 
Recall that the first phase of the algorithm consists in a double-sweep LexBFS.
Hence, it can be done in linear time.
Then, at every step of the second phase we must decide whether $H$ is a dominating set of $G$, that can be done in linear time.
If $H$ is not a dominating set, then we compute an extremity $x_i \notin N[H]$, that can also be done in linear time by Lemma~\ref{lem:next-extrem}.
We further compute an arbitrary shortest $x_ic$-path $P_i$, that can be done in linear time using BFS.
Overall, the second phase takes ${\cal O}(tm)$ time, with $t$ the number of extremities computed.
Finally, in the third phase, we need to execute ${\cal O}(\alpha)$ BFS for every shortest path $P_1,P_2,\ldots P_t$, that takes ${\cal O}(\alpha tm)$ time.

By Lemma~\ref{lem:lexbfs}, both $x_1$ and $x_2$ are also extremities.
We observe that $x_1x_2 \notin E$ (else, since $x_2 \in F(x_1)$, $x_1$ would be universal, thus contradicting either that $G$ is prime or $|V| \geq 3$).
By construction, $x_3,x_4,\ldots,x_t$ are pairwise nonadjacent, and they are also nonadjacent to both $x_1$ and $x_2$.
Altogether combined, we obtain that $x_1,x_2,\ldots,x_t$ are pairwise nonadjacent extremities.
As a result, $t \leq \alpha$.
It implies that the total runtime is in ${\cal O}(\alpha^2 m)$.
\end{proof}

\subsection{Exact computation}\label{sec:exact}

The following general result, of independent interest, is the cornerstone of Theorem~\ref{thm:exact}:

\begin{lemma}\label{lem:search-diam}
Let $u$ be a vertex in a prime graph $G=(V,E)$.
If $G$ has $q$ extremities, then we can compute in ${\cal O}(q m)$ time the value $\ell(u) = \max\{ e(x) \mid x \in N[u] \}$, and a $x \in N[u]$ of eccentricity $\ell(u)$.
This is ${\cal O}(\alpha m^{3/2})$ time if $G \in Ext_{\alpha}$.
\end{lemma}

The proof of Lemma~\ref{lem:search-diam} involves several cumbersome intermediate lemmas.
We postpone the proof of Lemma~\ref{lem:search-diam} to the end of this section, proving first our main result:

\begin{theorem}\label{thm:exact}
For every graph $G = (V,E)\in Ext_{\alpha}$, we can compute its diameter in deterministic ${\cal O}(\alpha^3m^{3/2})$ time.
\end{theorem}
\begin{proof}({\it Assuming Lemma~\ref{lem:search-diam}})
By Lemma~\ref{lem:red-mod}, we may assume $G$ to be prime.
Let us further assume that $|V|\geq 3$, and so that $G$ cannot have a universal vertex.

\smallskip
\noindent
{\bf Algorithm.}
We subdivide the procedure in three main phases, the two first of which being common to both Theorems~\ref{thm:approx} and~\ref{thm:exact}.
\begin{itemize}
    \item Let $x_1$ be the last vertex numbered in a LexBFS. 
    We execute a LexBFS with start vertex $x_1$. Let $x_2$ be the last vertex numbered in a LexBFS($x_1$).
    We compute a $c \in I(x_1,x_2)$ so that $d(c,x_1) = \left\lfloor d(x_1,x_2)/2\right\rfloor$.
    Let $P_1,P_2$ be shortest $x_1c$-path and $x_2c$-path, respectively.
    \item We set $H := P_1 \cup P_2$.
    While $H$ is not a dominating set of $G$, we compute a new extremity $x_i \notin N[H]$, an arbitrary shortest $x_ic$-path $P_i$, then we set $H := H \cup P_i$. 
    In what follows, we denote by $x_1,x_2,\ldots,x_t$ the extremities computed in the two first phases of the algorithm.
    Let $P_1,P_2,\ldots,P_t$ be the corresponding shortest paths, whose union equals $H$.
    \item Finally, for every $1 \leq i \leq t$, let $U_i \subseteq V(P_i)$ contain the $\min\{d(x_{i},c)+1,42\alpha-11\}$ closest vertices to $x_i$.
    Let $L_i := \max\{ \ell(u_i) \mid u_i \in U_i \}$.
    We output $L := \max_{1 \leq i \leq t}L_i$ as the diameter value.
\end{itemize}

\smallskip
\noindent
{\bf Complexity.}
The first phase of the algorithm can be done in ${\cal O}(m)$ time, and its second phase in ${\cal O}(tm)$ time, with $t$ the number of extremities computed.
During the third and final phase, we need to apply Lemma~\ref{lem:search-diam} ${\cal O}(\alpha)$ times for every shortest path $P_1,P_2,\ldots,P_t$.
Hence, the above algorithm runs in ${\cal O}(t\alpha qm)$ time, with $q$ the total number of extremities of $G$.
By Lemma~\ref{lem:num-extrem}, we have that $q={\cal O}(\alpha \sqrt{m})$.
See Theorem~\ref{thm:approx} for a proof that $t \leq \alpha$.
As a result, the above algorithm runs in ${\cal O}(\alpha^3m^{3/2})$ time.

\smallskip
\noindent
{\bf Correctness.}
Let us consider an arbitrary diametral pair $(u,v)$. 
Since $H$ is a dominating set of $G$, we have $N[u] \cap H \neq \emptyset$. 
Let $u^* \in N[u] \cap H$, and observe that $e(u^*) \geq diam(G) -1$. 
Then we claim that if $u^* \in V(P_i)$ for some $1 \leq i \leq t$, we must have $u^* \in U_i$. 
The proof is similar to what we did for Theorem~\ref{thm:approx}. 
Specifically, let us choose $\delta$ such that $G$ is $\delta$-hyperbolic.
By Lemma~\ref{lem:rad-hyp}, $e(c) \leq  rad(G) + 5\delta$.
Furthermore, if $u^* \in H$ satisfies $e(u^*) \geq diam(G) - 1$, then by Lemma~\ref{lem:diam-hyp}, $e(u^*) \geq 2rad(G) - 4\delta -2$.
In particular, $d(u^*,c) \geq rad(G) - 9\delta - 2$ (else, $e(u^*) \leq d(u^*,c) + e(c) \leq 2rad(G) - 4\delta - 3 < diam(G) - 1$).
For $1 \leq i \leq t$ such that $u^* \in V(P_i)$, since $P_i$ is a shortest $x_ic$-path of length at most $e(c) \leq rad(G) + 5\delta$, we get that $u^*$ must be one of the $(rad(G) + 5\delta)-(rad(G) - 9\delta - 2)+1= 14\delta+3$ closest vertices to $x_i$.
By Lemma~\ref{lem:hyp-ext-alpha}, $\delta \leq 3\alpha-1$, therefore $14\delta+3 \leq 42\alpha-11$.
In this situation, $L \geq L_i \geq \ell(u^*) = e(u) = diam(G)$. Combined with the trivial inequality $L \leq diam(G)$, it implies that $L = diam(G)$.
\end{proof}

The actual runtime of Theorem~\ref{thm:exact} is ${\cal O}(\alpha^2qm)$, where $q$ denotes the number of extremities.
By Lemma~\ref{lem:num-extrem}, this is in ${\cal O}(\alpha^3m)$ for bipartite graphs, ${\cal O}(\alpha^3\Delta m)$ for graphs with maximum degree $\Delta$ and in ${\cal O}(\alpha^5m)$ for triangle-free graphs. More generally, this is linear time for all graphs of $Ext_{\alpha}$ with bounded clique number.

\medskip
The remainder of this section is devoted to the proof of Lemma~\ref{lem:search-diam}.
The key idea here is that for every vertex $u$, there is an extremity $v \in F(u)$ such that $\max\{d(v,x) \mid x \in N[u]\} = \ell(u)$.
Therefore, in order to achieve the desired runtime for Lemma~\ref{lem:search-diam}, it would be sufficient to iterate over all extremities of $G$ that are contained in $F(u)$. 
However, due to our inability to compute all extremities in subquadratic time, we are bound to use Lemma~\ref{lem:next-extrem} for only computing some of these extremities.
Therefore, throughout the algorithm, we further need to grow some $u$-transitive set whose vertices must be carefully selected so that they can be discarded from the search space.
The next Lemmas~\ref{lem:clean} and~\ref{lem:clean-2} are about the construction of this $u$-transitive set.

\begin{lemma}\label{lem:clean}
Let $u$ and $v$ be vertices of a graph $G$ such that $v \in F(u)$, and let $X = \{ x \in N[u] \mid d(x,v) = e(u) \}$.
In ${\cal O}(m)$ time we can construct a set $Y$ where:
\begin{itemize}
    \item $v \in Y$; $Y$ is $u$-transitive;
    \item $d(y,x) \leq d(v,x)$ for every $y \in Y$ and $x \in X$.
\end{itemize}
\end{lemma}
\begin{proof}
The algorithm goes as follows:
\begin{enumerate}
    \item If $uv \in E$, then we output $Y = \{v\}$ and we halt. Otherwise, we set $Y_0 := V \setminus F(u), \ Y' := \emptyset$, and the algorithm continues to the next line.
    \item We construct a shortest-path tree $T$ rooted at $v$. Let $E_X \subseteq E(T)$ contain all the second edges $p_xq_x$ (starting from $v$) of the $vx$-paths in $T$, for $x \in X$.
    In what follows, we always assume that $p_x \in N(v)$. -- We observe later in the proof that for a vertex $w$ to satisfy $u \perp_w v$, it must intercept each edge of $E_X$ (see Claim~\ref{claim:intersect}). --
    \item We set $Y' := \{z \in F(u) \mid N(v) \cap I(u,v) \subseteq N(z) \}$. -- Note that every vertex $w \in F(u)$ such that $u \perp_w v$ belongs to this subset. However, we may also have some vertices currently in $Y'$ which do not satisfy this property. -- 
    \item We remove from $Y'$ every vertex which avoids at least one edge from $E_X$.
    \item Finally, we output $Y = Y_0 \cup Y'$. 
\end{enumerate}

\smallskip
\noindent
{\bf Correctness.}
The algorithm is trivially correct if $uv \in E$. Thus, from now on we assume $uv \notin E$.
In particular, $e(u) > 1$.
The output set $Y$ is the disjoint union of two subsets $Y_0$ and $Y'$, that we separately analyse in what follows. 

Recall that $Y_0 := V \setminus F(u)$.
In particular, for any $y \in Y_0$ and $x \in X$, we have $d(x,y) \leq d(u,y) + 1 \leq e(u) = d(v,x)$.
Similarly, by construction we have $d(x,y) \leq d(x,v)$ for every vertex $y$ of $Y'$ and every $x \in X$, that is because $y$ has in its closed neighbourhood some vertex on a shortest $vx$-path that is not equal to $v$ (namely, one of $p_x,q_x$ where $p_xq_x \in E_X$ and the $vx$-path in $T$ goes by this edge).

By Lemma~\ref{lem:transitive}, $Y_0$ is $u$-transitive.
Therefore in order to prove that $Y$ is $u$-transitive, it suffices to prove that for every $y \in Y'$ and $w \in V$ nonadjacent, $y \perp_w u \Longrightarrow w \in Y$.
For that, we start proving the following useful claim:
\begin{myclaim}\label{claim:intersect}
For $x \in X$, let $vp_x, p_xq_x$ be the two first edges of some shortest $vx$-path $P_x$.
There is no vertex of $V(P_x \setminus \{v,p_x,q_x\})$ in $N[F(u)]$.
\end{myclaim}

Indeed, suppose by contradiction the existence of such vertex $w'$.
Note that $d(w',x) = d(v,x) - d(w',v) \leq d(v,x) - 3$.
However, $d(w',u) \geq e(u) -1$ because $w' \in N[F(u)]$.
In particular, we have $d(w',x) \geq d(w',u) - 1 \geq e(u)-2 = d(v,x)-2$.
A contradiction.
Therefore, Claim~\ref{claim:intersect} is proved. $\diamond$

\smallskip
Suppose by contradiction the existence of nonadjacent vertices $y \in Y'$ and $w \in V \setminus Y$ such that $u \perp_w y$. 
In particular, $w \in F(u) \setminus Y'$.
There are two similar cases to be considered.
First, let us assume $N(v) \cap I(u,v) \not\subseteq N(w)$.
Let $p_u \in N(v) \cap I(u,v)$ be satisfying $p_u \notin N(w)$.
Note that $p_u \in N(y)$ and, since $y \in F(u)$, $p_u \in N(y) \cap I(u,y)$.
However, every shortest $uy$-path that contains $p_u$ is avoided by $w$, thus contradicting that $y \perp_w u$.
From now on, $N(v) \cap I(u,v) \subseteq N(w)$.
Then, let us consider an edge $p_xq_x \in E_X$ that is missed by $w$. 
This edge is on the $xv$-path in $T$, for some $x \in X$.
Let us call this path $P_x$. 
By Claim~\ref{claim:intersect}, $w$ misses $P_x \setminus \{v\}$.
However, since $u,y \in N[P_x \setminus \{v\}]$, vertices $u$ and $y$ must be in the same connected component of $G \setminus N[w]$.
A contradiction.

Finally, we have $v \in Y' \subseteq Y$ by construction.

\smallskip
\noindent
{\bf Complexity.}
We can construct the shortest-path tree $T$ in linear time, simply by running a BFS.
Then, we can construct the subset $E_X$ of tree edges in additional ${\cal O}(n)$ time by using a standard dynamic programming approach on $T$ (bottom-up, from the nodes in $X$ to the root $v$).
For convenience, let us define $Z = \{z \in F(u) \mid N(v) \cap I(u,v) \subseteq N(z) \}$, that is a superset of the final subset $Y'$.
It can be constructed in linear time, simply by scanning the adjacency lists of all vertices of $N(v) \cap I(u,v)$.
Let also $E_Z$ be the subset of all edges $pq \in E_X$ such that $p \in N(v) \setminus I(u,v)$.
Indeed, we can observe that a vertex $z \in Z$ intercepts any edge of $E_X \setminus E_Z$.
Therefore, in order to decide whether such a $z$ intercepts all edges of $E_X$, it suffices to check whether that is the case for all edges of $E_Z$.

For that, we create two sets ${\cal P}$ and ${\cal Q}$ so that, for every $pq \in E_Z$, where $p \in N(v)$, we have $p \in {\cal P}$ and $q \in {\cal Q}$.
We totally order the vertices of ${\cal P}$, and we reorder the partial adjacency lists $N[z] \cap {\cal P}, \ z \in Z$, according to this total ordering ({\it i.e.}, for every $z \in Z$, we store an ordered list whose elements are exactly $N[z] \cap {\cal P}$). It takes linear time.
Then, we totally order ${\cal Q}$ so that, for any $p,p' \in {\cal P}$, if $p$ is ordered before $p'$, then all children nodes of $p$ in ${\cal Q}$ are ordered before the children nodes of $p'$ in ${\cal Q}$ (with respect to the tree $T$). 
It can be done in ${\cal O}(n)$ time, {\it e.g.}, if we scan the children nodes of every $p \in {\cal P}$, in order.
We reorder the partial adjacency lists $N(z) \cap {\cal Q}, \ z \in Z$, according to the total ordering of ${\cal Q}$. It also takes linear time.
Finally, for every $p \in {\cal P}$, let $\iota(p)$ be the largest $q \in {\cal Q}$ (according to the total order) so that $pq \in E_Z$.
We consider each $z \in Z$ sequentially, and we do as follows:
\begin{itemize}
    \item We scan the ordered lists $N(z) \cap {\cal Q}, N[z] \cap {\cal P}$ and ${\cal Q}$. Let $q,p$ and $q'$ be the next vertices to be read in all of the three lists. If $q=q'$ then we move forward in both $N(z) \cap {\cal Q}$ and ${\cal Q}$. Otherwise, $q'$ is ordered before $q$. Let $p'$ be the parent node of $q'$ in $T$ (it can be computed in ${\cal O}(1)$ time assuming a standard pointer structure for the tree $T$). We move forward in $N[z] \cap {\cal P}$ until the next vertex $p$ to be read either equals $p'$ or is ordered after it. If $p \neq p'$, then the edge $p'q' \in E_Z$ is avoided by vertex $z$ and we stop, removing $z$ from $Y'$. Else, we move forward in $N(z) \cap {\cal Q}$ until the next vertex to be read is ordered after $\iota(p)$. We access to the position of $\iota(p)$ in ${\cal Q}$ in ${\cal O}(1)$ time (using a pointer) and then we move forward in ${\cal Q}$ from one position.
\end{itemize}
Doing so, the total number of operations for any $z \in Z$ is upper bounded by an ${\cal O}(|N(z)|)$. As a result, this whole step takes linear time.
\end{proof}

We now complete Lemma~\ref{lem:clean}, as follows:

\begin{lemma}\label{lem:clean-2}
Let $u$ and $v$ be vertices of a graph $G$ such that $v \in F(u)$, and $d(x,v) \leq e(u)$ for every $x \in N[u]$.
In ${\cal O}(m)$ time we can construct a set $S'$ where:
\begin{itemize}
    \item $v \in S'$; $S'$ is $u$-transitive;
    \item $d(s,x) \leq e(u)$ for every $s \in S'$ and $x \in N[u]$.
\end{itemize}
\end{lemma}
\begin{proof}
Let $X = \{ x \in N[u] \mid d(v,x) = e(u) \}$ and let $X' = N[u] \setminus X = \{x' \in N[u] \mid d(v,x') = e(u)-1\}$.
Let $Y$ be as in Lemma~\ref{lem:clean} applied to $u$ and $v$.
We obtain $S'$ by removing some vertices from $Y$:
\begin{enumerate}
    \item For every $y \in Y$, we compute $d(y,X') = \min\{ d(y,x') \mid x' \in X'\}$.
    \item Then, we consider each $y \in Y$ sequentially:
    \begin{enumerate}
        \item If $d(y,X') \leq e(u) - 2$, then we keep vertex $y$ in $S'$.
        \item Else, if $d(y,X') \geq e(u)$, then we do {\em not} keep this vertex in $S'$. -- Intuitively, this is because $y$ may be potentially at distance $e(u)+1$ from some vertex of $X'$. --
        \item Else, $d(y,X') = e(u)-1$. We keep $y$ in $S'$ iff $\bigcup \{N(v) \cap I(x',v) \mid x' \in X' \} \subseteq N(y)$. 
    \end{enumerate}
\end{enumerate}

\smallskip
\noindent
{\bf Correctness.}
We need to prove that $S'$ is $u$-transitive.
Recall that $S' \subseteq Y$, that is a $u$-transitive set.
By contradiction, let $s \in S'$ and $y \in V$ nonadjacent satisfy $s \perp_y u$ but $y \notin S'$.
In particular, $y \in Y \setminus S'$.
There are two cases:
\begin{itemize}
    \item Case $d(y,X') > d(s,X')$. We consider a $x' \in X'$ closest to $s$ and a corresponding shortest $sx'$-path $P'$. Since $d(s,x') < d(y,x')$ and $sy \notin E$, $y$ misses $P'$. But, since $x'u \in E$, vertices $u$ and $s$ must be in the same connected component of $G \setminus N[y]$. A contradiction.
    \item Otherwise, $d(y,X') \leq d(s,X')$. Since we have $d(y,X') > e(u) -2$ (else, $y \in S'$) and $d(s,X') < e(u)$ (else, $s \notin S'$), the only possibility is that $d(y,X') = d(s,X') = e(u)-1$. Since we assume $y \notin S'$, there exists a $w \in \bigcup \{N(v) \cap I(x',v) \mid x' \in X' \}$ nonadjacent to $y$. For some $x' \in X'$, there is a shortest $x'v$-path $P'$ going by $w$. Then, $y$ misses $P' \setminus \{v\}$. But, since $u,s \in N[P' \setminus \{v\}]$, again we have that the vertices $u$ and $s$ must be in the same connected component of $G \setminus N[y]$. A contradiction.
\end{itemize}
As a result, we proved that $S'$ is indeed $u$-transitive.

Next, we prove that $v \in S'$.
Indeed, $v \in Y$. 
Furthermore, $d(v,X') = e(u) -1$ and $v$ is adjacent to all of $\bigcup \{N(v) \cap I(x',v) \mid x' \in X' \}$.

What finally remains to prove is that for every $s \in S'$ and $x \in N[u]$ we have $d(s,x) \leq e(u)$.
For that, let $s \in S'$ be arbitrary.
If $d(s,X') \leq e(u)-2$, then (since the vertices of $N[u]$ are pairwise at distance at most two) $d(s,x) \leq e(u)$ for every $x \in N[u]$. 
From now on, we assume $d(s,X') = e(u) -1$.
As $s \in Y$ we have $d(s,x) \leq d(v,x) \leq e(u)$ for every $x \in X$.
Now, for any $x' \in X'$, let $w \in \bigcup \{N(v) \cap I(x',v) \mid x' \in X' \}$ be on a shortest $vx'$-path.
We have $d(s,x') \leq 1 + d(w,x') = d(v,x') = e(u) -1$.

\smallskip
\noindent
{\bf Complexity.}
We can compute $X$ and $X'$ in linear time by simply running a BFS rooted at $v$.
The initial subset $Y$ can be also constructed in linear time.
Then, we can compute the distances $d(w,X')$, for every vertex $w \in V$, by running a BFS rooted at $X'$ (equivalently, we replace $X'$ by a single new vertex with neighbour-set $N(X')$ and then we run a BFS rooted at this new vertex).
It takes linear time.
Doing so, since all vertices of $X'$ are at equal distance $e(u)-1$ from vertex $v$, we also computed the subset $\bigcup \{N(v) \cap I(x',v) \mid x' \in X' \}$.
Then, we can compute all vertices $y \in S'$ such that $d(y,X') = e(u)-1$ in linear time, simply by scanning the adjacency lists of all vertices in $\bigcup \{N(v) \cap I(x',v) \mid x' \in X' \}$.
\end{proof}
Using Lemma~\ref{lem:clean-2} we end up the section proving Lemma~\ref{lem:search-diam}:
\begin{proof}[Proof of Lemma~\ref{lem:search-diam}]
We may assume that $|V| \geq 3$ and therefore (since $G$ is prime) that $u$ is not a universal vertex.
We search for a $v \in V$ at a distance $\ell(u)$ from some vertex of $N[u]$.
Note that we have $e(u) \leq \ell(u) \leq e(u)+1$.
In particular we only need to consider the vertices $v \in F(u)$.
We next describe an algorithm that iteratively computes some pairs $(v_0,S_0),(v_1,S_1),\ldots$ such that, for any $i \geq 0$: {\tt(i)} $v_i \in F(u)$ is an extremity; and {\tt(ii)} $S_i$ is $u$-transitive. We continue until either $d(v_i,x) = e(u) +1$ for some $x \in N(u)$ or $F(u) \subseteq S_i$. 

If $i=0$ then, let $v_0$ be the last vertex numbered in a LexBFS($u$), that is an extremity by Lemma~\ref{lem:lexbfs}. 
Otherwise ($i > 0$), let $v_i$ be the output of Lemma~\ref{lem:next-extrem} applied to $u$ and $S = S_{i-1}$.
Furthermore, being given the extremity $v_i$, let $S_i'$ be computed as in Lemma~\ref{lem:clean} applied to $u$ and $v_i$.
Let $S_i = S_i' \cup S_{i-1}$ (with the convention that $S_{-1} = \emptyset$).
We stress that $S_i$ is $u$-transitive since it is the union of two $u$-transitive sets.

By Lemma~\ref{lem:clean-2}, all vertices in $S_i$ can be safely discarded since they cannot be at distance $e(u)+1$ from any vertex of $N[u]$. Furthermore, since for every $i$ we have $v_i \in S_i \setminus S_{i-1}$, the sequence $(F(u) \setminus S_{i-1})_{i \geq 0}$ is strictly decreasing with respect to set inclusion.
Each step $i$ takes linear time, and the total number of steps is bounded by the number of extremities in $F(u)$. 
For a graph within $Ext_{\alpha}$, this is in ${\cal O}(\alpha\sqrt{m})$ according to Lemma~\ref{lem:num-extrem}.
\end{proof}

\section{Extensions}\label{sec:extensions}

Our algorithmic framework in Sec.~\ref{sec:exact} can be refined in several ways.
We present such refinements for the larger class of all graphs having a dominating target of bounded cardinality.
Some general results are first discussed in Sec.~\ref{sec:gal-dom-target} before we address the special case of chordal graphs in Sec.~\ref{sec:chordal}.

\subsection{More results on dominating targets}\label{sec:gal-dom-target}

We start with the following observation:

\begin{lemma}\label{lem:red-mod-2}
If a graph $G$ contains a dominating target of cardinality at most $k$, then so does its quotient graph $G'$.
\end{lemma}
\begin{proof}
Given a module $M$ of a graph $G$, let $H$ be obtained from $G$ by removing all but one vertex of $M$. 
We claim that if $G$ has a dominating target of cardinality at most $k$, then so does $H$. 
Indeed, given such a dominating target $D$, it suffices to keep all vertices from $D \setminus M$, and to add to the latter the unique vertex of $M \cap V(H)$ if and only if $M \cap D \neq \emptyset$. 
The above proves the lemma because the quotient subgraph $G'$ is obtained by repeatedly applying this above-described operation to pairwise disjoint modules.
\end{proof}

By Lemma~\ref{lem:red-mod-2}, we may only consider in what follows {\em prime} graphs with a dominating target of bounded cardinality.
By $K_{1,t}$, we mean the star with $t$ leaves.
A graph is $K_{1,t}$-free if it has no induced subgraph isomorphic to $K_{1,t}$.
\begin{lemma}\label{lem:k1t}
Every $K_{1,t}$-free graph $G$ with a dominating target of cardinality at most $k$ belongs to $Ext_{k(t-1)}$.
\end{lemma}
\begin{proof}
By Lemma~\ref{lem:red-mod-2}, we may assume the graph $G$ to be prime.
Suppose by contradiction the existence of $k(t-1)+1$ pairwise nonadjacent extremities of $G$.
Let $D$ be dominating target of $G$ of cardinality at most $k$.
By Lemma~\ref{lem:extrem-close-to-target}, every of these $k(t-1)+1$ extremities is in $N[D]$.
But then, by the pigeonhole principle, some vertex of $D$ is adjacent to at least $t$ such extremities, that results in the existence of an induced $K_{1,t}$ in $G$.
A contradiction.
\end{proof}

\begin{corollary}\label{cor-k1t}
We can compute the diameter of $K_{1,t}$-free graphs with a dominating target of cardinality at most $k$ in deterministic ${\cal O}((kt)^3m^{3/2})$ time.
\end{corollary}

Let us now consider graphs with a dominating target of cardinality at most $k$ and bounded maximum degree $\Delta$. 
These graphs are $K_{1,\Delta+1}$-free and therefore, according to Corollary~\ref{cor-k1t}, we can compute their diameter in deterministic ${\cal O}((k\Delta)^3m^{3/2})$ time.
We improve this runtime to quasi linear, while also decreasing the dependency on $\Delta$, namely:

\begin{theorem}\label{thm:diam-dom}
For every $G = (V,E)$ with a dominating target of cardinality at most $k$ and maximum degree $\Delta$, we can compute its diameter in deterministic ${\cal O}(k^3\Delta m \log{n})$ time.
\end{theorem}

Since under SETH we cannot compute the diameter of graphs with a dominating edge in subquadratic time~\cite{Duc21c}, the dependency on $\Delta$ in Theorem~\ref{thm:diam-dom} is conditionally optimal.

Being given a vertex $c$ of small eccentricity, our main difficulty here is to compute efficiently (as we did for Theorem~\ref{thm:exact}) a small number of shortest-paths starting from $c$ whose union is a dominating set of $G$. This could be done by computing a dominating target of small cardinality. However, our framework in the prior Sec.~\ref{sec:framework} only allows us to compute extremities, and not directly a dominating target.
Our strategy consists in including in some candidate subset all the neighbours of the extremities that are computed by our algorithm.
By using Lemma~\ref{lem:extrem-close-to-target}, we can bound the size of this candidate subset by an ${\cal O}(k\Delta)$.
Then, by using a greedy set cover algorithm, we manage to compute from this candidate subset a set of ${\cal O}(k\log{n}$) shortest-paths that cover all but ${\cal O}(k\Delta \log{n})$ vertices.
We apply the prior techniques of Sec.~\ref{sec:alg} to all these shortest-paths (calling upon Lemma~\ref{lem:search-diam}), while for the ${\cal O}(k\Delta \log{n})$ vertices that they miss we compute their eccentricities directly.

\begin{proof}[Proof of Theorem~\ref{thm:diam-dom}]
By Lemma~\ref{lem:red-mod-2}, we may assume $G$ to be prime.
We proceed as follows:
\begin{enumerate}
    \item We compute a vertex $c$ of eccentricity at most $rad(G) + 15k - 5$.
    \item Then, we set $H=\{c\}, \ X = \emptyset$.
    While $H$ is not a dominating set of $G$, we compute an extremity $x_i \notin N[H]$, we add in $X$ {\em all} vertices of $N[x_i]$ and, for every $y \in N[x_i]$, we add to $H$ some arbitrary shortest $yc$-path $P_y$.
    \item Let ${\cal S} = \{ N[P_x] \mid x \in X \} \cup \{ N[v] \mid v \in V \}$. 
    We apply a greedy set cover algorithm in order to extract from ${\cal S}$ a sub-family ${\cal S}'$ so that $\bigcup {\cal S'} = V$.
    Specifically, we set ${\cal S}' := \emptyset, \ U := V$, where $U$ represents the uncovered vertices.
    While $U \neq \emptyset$, we add to ${\cal S}'$ any subset  $S \in {\cal S}$ such that $|S \cap U|$ is maximized, and then we set $U := U \setminus S$.
    \item Let $A = \{ x \in X \mid N[P_x] \in {\cal S}' \}$ and let $B = \{ v \in V \setminus A \mid N[v] \in {\cal S}' \}$.
    \begin{enumerate}
        \item For every $x \in A$, let $W_x$ contain the $\min\{d(x,c)+1,42k-11\}$ closest vertices to $x$ in $P_x$. For each $w \in W_x$, we compute $\ell(w) = \max\{e(w') \mid w' \in N[w]\}$.
        \item For every $v \in B$, we directly compute the eccentricities of {\em all} vertices in $N[v]$.
    \end{enumerate}
    \item We output the maximum eccentricity computed as the diameter value.
\end{enumerate}

\smallskip
\noindent
{\bf Correctness.}
In order to prove correctness of this above algorithm, let $u \in V$ be such that $e(u) = diam(G)$.
By construction, ${\cal S}'$ covers $V$, and therefore, either there exists a $x \in A$ such that $u \in N[P_x]$, or there exists a $v \in B$ such that $u \in N[v]$.
In the latter case, we computed the eccentricities of all vertices in $N[v]$, including $e(u) = diam(G)$.
Therefore, we only need to consider the former case.
Specifically, let $u^* \in V(P_x) \cap N[u]$.
To prove correctness of the algorithm, it suffices to prove that $u^* \in W_x$.
The proof that it is indeed the case is identifical to that of Theorem~\ref{thm:exact}.
Indeed, since $G$ has a dominating target of cardinality at most $k$, it contains an additive tree $(3k-1)$-spanner~\cite{KKM01} and so -- the same as graphs in $Ext_k$ --, it is $(3k-1)$-hyperbolic~\cite{CDEH+08,DrK14}.

\smallskip
\noindent
{\bf Complexity.}
By Lemma~\ref{lem:rad-hyp}, we can compute a vertex $c$ of eccentricity at most $rad(G)+15k-5$ in linear time.

Then, during the second phase of the algorithm, we claim that each step can be done in ${\cal O}(\Delta m)$ time.
Indeed, while $H$ is not a dominating set of $G$, a new extremity $x_i \notin N[H]$ can be computed by applying Lemma~\ref{lem:next-extrem}.
Furthermore, we can compute the shortest paths $P_y$, for $y \in N[x_i]$, by executing at most $\Delta+1$ BFS.
Overall, the runtime of this second phase is in ${\cal O}(t\Delta m)$, with $t$ the number of extremities computed.
We claim that $t \leq k$.
Indeed, let $D$ be a fixed (but unknown) dominating target of cardinality at most $k$.
Let $x_1,x_2,\ldots,x_t$ denote all the extremities computed during this second phase.
By Lemma~\ref{lem:extrem-close-to-target}, each extremity $x_i$ computed is in $N[D]$.
In order to prove that there are at most $k$ extremities computed, it suffices to prove that $N[x_i] \cap N[x_j] \cap D = \emptyset$ for every $i < j$.
Suppose by contradiction that there exists a vertex $v \in D$ such that $v \in N[x_i] \cap N[x_j]$.
At step $i$, we add $v$ in $X$, and therefore from this point on $x_j \in N[H]$.
It implies that we cannot select $x_j$ at step $j$, a contradiction.
Hence, the total runtime of this second phase is in ${\cal O}(k\Delta m)$.

In order to bound the runtime of the third phase of the algorithm, we first need to bound the minimum number of subsets of ${\cal S}$ needed in order to cover $V$.
We claim that it is no more than $2k$.
Indeed, as before let $D$ be a fixed (but unknown) dominating target of cardinality at most $k$.
For each $x \in D \cap X$, we select the set $P_x$. 
For each $v \in D \setminus X$, we select the sets $N[v]$ and $P_{x'}$, for some arbitrary $x' \in X$ such that $v \in N(P_{x'})$. 
The claim follows since we assume $D$ to be a dominating target.
It implies that the greedily computed sub-family ${\cal S}'$ has its cardinality in ${\cal O}(k\log{n})$. 
The cumulative size of all subsets in ${\cal S}$ is at most $|X|n + 2m \leq k(\Delta+1) n +2m$.
Furthermore, the greedy set cover algorithm runs in $|{\cal S}'| = {\cal O}(k\log{n})$ steps.
As a result, the total runtime for this third phase is in ${\cal O}(k\log{n} \cdot (k\Delta n + m)) = {\cal O}(k^2\Delta m \log{n})$.

Lastly, during the fourth and final phase, we need to apply Lemma~\ref{lem:search-diam} $\sum_{x \in A}|W_x| = {\cal O}(k|A|)$ times, and we need to execute $\sum_{v \in B}|N[v]| = {\cal O}(\Delta |B|)$ BFS.
Each call to Lemma~\ref{lem:search-diam} takes ${\cal O}(qm)$ time, with $q$ the number of extremities.
By Lemma~\ref{lem:extrem-close-to-target}, $q = {\cal O}(k\Delta)$.
Furthermore, we recall that $|A|+|B| \leq |{\cal S}'| = {\cal O}(k\log{n})$.
Therefore, the total runtime for this phase -- and also for the whole algorithm -- is in ${\cal O}(k^3\Delta m \log{n})$.
\end{proof}

\subsection{Chordal graphs}\label{sec:chordal}

By~\cite[Theorems $6$ \& $9$]{Duc21c}, we can decide in linear time the diameter of chordal graphs with a dominating pair (resp., with a dominating triple) if the former value is at least $4$ (resp., at least $10$). 
It was also asked in~\cite{Duc21c} whether for every $k \geq 4$, there exists a threshold $d_k$ such that the diameter of chordal graphs with a dominating target of cardinality at most $k$ can be decided in truly subquadratic time if the former value is at least $d_k$.
We first answer to this question in the affirmative:

\begin{theorem}\label{thm:chordal-domt}
If $G=(V,E)$ is chordal, with a dominating target of cardinality at most $k$, and such that $diam(G) \geq 4$, then we can compute its diameter in deterministic ${\cal O}(km)$ time.
\end{theorem}

Roughy, we lower the runtime of Theorem~\ref{thm:chordal-domt} to linear by avoiding calling upon Lemma~\ref{lem:search-diam}.
More specifically, we prove that whenever $diam(G) \geq 4$ there is always one of the $\leq k$ extremities which we compute whose eccentricity equals the diameter.
The analysis of Theorem~\ref{thm:chordal-domt} is involved, as it is based on several nontrivial properties of chordal graphs~\cite{ChD94,CDDH+01,CDK03,LaS83,RTL76}.
We first introduce all the properties that we need.
In what follows, recall that a vertex is {\em simplicial} if its neighbourhood induces a clique.

\begin{lemma}[\cite{RTL76}]\label{lem:peo}
If a graph $G$ is chordal, then the vertex last numbered in a LexBFS is simplicial.
\end{lemma}

\begin{lemma}[\cite{CDDH+01}]\label{lem:lexbfs-chordal}
Let $u$ be the vertex of a chordal graph $G$ last numbered in a LexBFS, and let $x,y$ be a pair of vertices such that $d(x,y) = diam(G)$. If $e(u) < diam(G)$ then $e(u)$ is even, $d(u,x) = d(u,y) = e(u)$ and $e(u) = diam(G)-1$.
\end{lemma}

Being given a subset $S$ and a vertex $v \notin S$, let $d(v,S) = \min\{d(v,s) \mid s \in S\}$.
Let $P(v,S) = \{ s \in S \mid d(v,s) = d(v,S)\}$ be the projection of $v$ on $S$.

\begin{lemma}[\cite{ChD94}]\label{lem:chordal-inclusion-proj}
In a chordal graph $G$, for any clique $C$ and adjacent vertices $u,v \notin C$ the metric projections $P(u,C)$ and $P(v,C)$ are comparable, {\it i.e.}, either $P(u,C) \subseteq P(v,C)$ or $P(v,C) \subseteq P(u,C)$.
\end{lemma}

Finally, we recall the following tight relation between diameter and radius:

\begin{lemma}[\cite{LaS83}]\label{lem:center-diam-chordal}
For every chordal graph $G$, $2rad(G) - 2 \leq diam(G) \leq 2rad(G)$.
\end{lemma}

\begin{proof}[Proof of Theorem~\ref{thm:chordal-domt}]
By Lemma~\ref{lem:red-mod-2}, we may assume $G$ to be prime. 
We proceed as follows:
\begin{itemize}
    \item We pick a start vertex $c$ so that $e(c) = rad(G)$.
    \item We set $H=\{c\}$. While $H$ is not a dominating set of $G$, we compute an extremity $x_i \notin N[H]$ and we set $H := H \cup P_i \cup N(x_i)$, where $P_i$ denotes some arbitrary shortest $cx_i$-path. Let $x_1,x_2,\ldots,x_t$ denote all the extremities computed.
    \item We output $\max\{e(x_i) \mid 1 \leq i \leq t\}$ as the diameter value for $G$.
\end{itemize}

\smallskip
\noindent
{\bf Complexity.}
For chordal graphs, a central vertex can be computed in linear time~\cite{ChD94}.
The two last phases of the algorithm run in ${\cal O}(tm)$ time, with $t$ the number of extremities computed.
We claim that $t \leq k$.
Indeed, let $D$ be a fixed (but unknown) dominating target with at most $k$ vertices.
By Lemma~\ref{lem:extrem-close-to-target}, $x_1,x_2,\ldots,x_t \in N[D]$.
Furthermore, for every $1 \leq i < j \leq t$, since at step $i$ we fully added $N[x_i]$ to $H$, $N[x_i] \cap N[x_j] = \emptyset$ (otherwise, $x_j \in N[H]$ cannot be selected at ulterior step $j > i$).
In particular, $N[x_i] \cap N[x_j] \cap D = \emptyset$, and so we proved as claimed that $t \leq |D| \leq k$.
Overall, the total runtime is in ${\cal O}(km)$.

\smallskip
\noindent
{\bf Correctness.}
We prove in what follows that $diam(G) = \max\{e(x_i) \mid 1 \leq i \leq t\}$.
Suppose by contradiction $diam(G) > \max\{e(x_i) \mid 1 \leq i \leq t\}$.
Recall that $diam(G) \geq 4 > 2$.
It implies the existence of some vertex $u \notin N[c]$ such that $e(u) = diam(G)$.
In what follows, let $i_0$ be the least index $i$ such that there is a vertex of maximum eccentricity in $\left( N[P_i] \cup N^2[x_i]\right) \setminus N[c]$.
Let also $u \in \left(N[P_{i_0}] \cup N^2[x_{i_0}]\right) \setminus N[c]$ satisfy $e(u) = diam(G)$, and let $v \in F(u)$.

We need to revisit the proof of Lemma~\ref{lem:next-extrem} in order to derive a contradiction.
Specifically, let $S = \bigcup_{j < i_0} \left(N[P_j] \cup N^2[x_j]\right)$.
In order to compute $x_{i_0}$, we first execute a LexBFS($c$), thus outputting a LexBFS ordering $\sigma$.
Then, let $y_{i_0}$ be the vertex of $V \setminus S$ with the smallest LexBFS number, and let $S'$ denote the subset of all vertices in $S$ that are numbered after $y_{i_0}$.
We extract $x_{i_0}$ from some module $M$ of $G \setminus S'$, $c \notin M$, whose vertices have the smallest LexBFS numbers amongst the vertices of $V \setminus S'$ (in particular, $y_{i_0} \in M$).

By minimality of index $i_0$ we have $u \notin S$.
Therefore, either $\sigma^{-1}(u) > \sigma^{-1}(x_{i_0})$ or $u \in M$.
Similarly, if $v \in S$, then by minimality of $i_0$ we have $v \in N[c]$, and therefore $\sigma^{-1}(v) > \sigma^{-1}(x_{i_0})$ because we have $d(c,v) < d(c,x_{i_0})$.
As a result, either $\sigma^{-1}(v) > \sigma^{-1}(x_{i_0})$ or $v \in M$.
Overall, it implies that $u,v \in V \setminus S'$.

Furthermore, by Lemma~\ref{lem:peo} we obtain $G \setminus S'$ from $G$ by iteratively removing some simplicial vertices, and therefore $G \setminus S'$ is an isometric subgraph of $G$.
In particular, $diam(G \setminus S') = diam(G) = d(u,v)$.

\begin{myclaim}
$u,v \notin M$.
\end{myclaim}
Indeed,we cannot have $u,v \in M$ since otherwise, $d(u,v) \leq 2$. 
But if $M \cap \{u,v\} \neq \emptyset$, then since $M$ is a module of $G \setminus S'$, all vertices of $M$ have maximum eccentricity in $G \setminus S'$.
In particular, $e(x_{i_0}) = diam(G)$, a contradiction.
Therefore, we proved as claimed $u,v \notin M$. $\diamond$

\smallskip
\noindent
In this situation, we may further restrict ourselves to the isometric subgraph $G_{i_0}$ induced by $x_{i_0}$ and all vertices numbered before $x_{i_0}$.
Since we still have $u,v \in V(G_{i_0})$, $diam(G_{i_0}) = diam(G) = d(u,v)$.
By construction, $x_{i_0}$ is the last vertex numbered in a LexBFS in $G_{i_0}$.
Since we suppose $e(x_{i_0}) < diam(G)$ and $G_{i_0}$ is chordal, the following results follow from Lemma~\ref{lem:lexbfs-chordal}: $e(x_{i_0}) = diam(G) - 1 = d(u,x_{i_0}) = d(v,x_{i_0})$.

In this situation, $u,v \notin N^2[x_{i_0}]$ because $d(u,x_{i_0}) = d(v,x_{i_0}) = diam(G) - 1 \geq 3$.
In particular, we have $u \in N[P_{i_0}]$.
Since we suppose $u \notin N(c)$, we get $d(x_{i_0},u) \leq d(x_{i_0},c)$.
As a result, $d(x_{i_0},u) = d(x_{i_0},c) = diam(G)-1$. 

\smallskip
We make the following useful observation for what follows: 

\begin{myclaim}\label{claim:rad-1}
$diam(G) = rad(G)+1$.
\end{myclaim}
Indeed, $rad(G) = e(c) \geq d(x_{i_0},c) = diam(G)-1$.
If $rad(G) > diam(G) - 1$, then $rad(G) = diam(G)$ and all vertices ({\it e.g.}, $x_{i_0}$) have maximum eccentricity. 
A contradiction. $\diamond$

\smallskip
We end up deriving a contradiction as we can now prove that:
\begin{myclaim}\label{claim:rad-2}
$rad(G) = 2$.
\end{myclaim}
Indeed, by Lemma~\ref{lem:center-diam-chordal} $diam(G) \geq 2rad(G) -2$.
Therefore, $rad(G)+1 \geq 2rad(G)-2$, that implies $rad(G) \leq 3$.
Suppose $rad(G) = 3$.
Since $x_1 \in F(c)$, we obtain $e(x_1) \geq e(c) = rad(G) = diam(G)-1$.
We cannot have $e(x_1) > e(c)$ because otherwise we obtain $e(x_1) = diam(G)$.
Hence, $e(x_1) = e(c) = 3$, that is an odd number. 
But again by Lemma~\ref{lem:lexbfs-chordal}, if the last vertex output by LexBFS in a chordal graph has odd eccentricity, then its eccentricity equals the diameter.
It implies that in fact $e(x_1) = diam(G)$.
A contradiction. $\diamond$

\smallskip
This above Claim~\ref{claim:rad-2} implies $diam(G) = 2rad(G) = 4$, therefore $diam(G) = rad(G)+2$, that contradicts Claim~\ref{claim:rad-1}.
\end{proof}

We complete the above Theorem~\ref{thm:chordal-domt} with an improved algorithm for computing the diameter of chordal graphs with bounded asteroidal number.
For that, we combine the algorithmic scheme of Theorem~\ref{thm:chordal-domt} with a previous approach from~\cite{Duc21c}.
More specifically, for the special case $rad(G)=2$, we considerably revisit a prior technique from~\cite[Proposition 2]{Duc21c}, for split graphs, so that is also works on diameter-three chordal graphs.

\begin{theorem}\label{thm:chordal}
For every chordal graph $G = (V,E)$ of asteroidal number at most $k$, we can compute its diameter in deterministic ${\cal O}(km)$ time.
\end{theorem}
\begin{proof}
By Lemma~\ref{lem:red-mod}, we may assume $G$ to be prime. 
Let us further assume $|V| \geq 3$, so that there is no universal vertex.
We pick a start vertex $c$ so that $e(c) = rad(G)$.
It can be computed in linear time since $G$ is chordal~\cite{ChD94}.

Let us first assume that $rad(G) \geq 3$.
By Lemma~\ref{lem:center-diam-chordal}, $diam(G) \geq 2rad(G)-2 \geq 4$.
By Lemma~\ref{lem:dom-target}, $G$ has a dominating target of cardinality at most $k$.
Therefore, we are done applying Theorem~\ref{thm:chordal-domt}.

From now on, we assume that $rad(G) \leq 2$.
In fact, since there is no universal vertex, $rad(G) = 2$.
Let $x_1$ be the last vertex numbered in a LexBFS($c$).
\begin{itemize}
    \item If $e(x_1) = 4 = 2rad(G)$, then clearly $e(x_1) = diam(G)$.
    \item Similarly, if $e(x_1) = 3$, then since $e(x_1)$ is odd, by Lemma~\ref{lem:lexbfs-chordal} $e(x_1) = diam(G)$.
\end{itemize}
From now on, $e(x_1) = 2$.
By Lemma~\ref{lem:peo}, $x_1$ is a simplicial vertex.
Thus $C = N(x_1)$ is a dominating clique of $G$, which implies $diam(G) \in \{2,3\}$.
Let $S = V \setminus N[x_1]$.
Note that any vertex of eccentricity $3$ must be in $S$.

\smallskip
We next define a specific type of orderings over the vertices in $S$, that is the cornerstone of our algorithmic procedure for this case.
Specifically, for every $v \in S$, let $\ell(v) = |C \cap N(v)|$.
Let $N(v) \cap S = (u_1,u_2,\ldots,u_d)$ be ordered by non-increasing $\ell(\cdot)$ values, and let $L(v) = (\ell(v),\ell(u_1),\ell(u_2),\ldots,\ell(u_d))$.
An $L$-ordering is a total ordering of $S$ by non-decreasing $L$-values -- according to the lexicographic ordering.

\smallskip
Our algorithm goes as follows:
\begin{enumerate}
    \item We compute an $L$-ordering.
    \item We scan the vertices in an $L$-ordering of $S$.
    For each vertex $x$, we discard all its neighbours that appear later in the ordering.
    Doing so, we obtain an independent set $S^* \subseteq S$.
    \item We scan the vertices in an $L$-ordering of $S^*$.
    Let us introduce the following binary relation $\prec_N$ over $S^*$.
    Namely, $u \prec_N v$ if and only if $N(y) \cap C \subseteq N(v)$ for every $y \in N[u] \cap S, \ y \notin N(v)$.
    For each vertex $x$, if $x$ has not been discarded at an earlier step (when we processed some other vertex that appears before $x$ in the $L$-ordering), then we discard all $v \in S^*$ such that $x \prec_N v$ that appear later in the ordering.
    \item We output the maximum eccentricity amongst the remaining vertices of $S^*$ as the diameter of $G$.
\end{enumerate}

\smallskip
\noindent
{\bf Correctness.}
Besides the initial computation of an $L$-ordering, this above algorithm has two main phases.
One consists in reducing $S$ to an independent set $S^*$.
The other consists in further pruning out $S^*$ by using the binary relation $\prec_N$.

Correctness of the first phase directly follows from the following claim:
\begin{myclaim}\label{claim:lex-chordal}
If $u,v \in S$ are adjacent so that $L(u) \trianglelefteq L(v)$, then $e(u) \geq e(v)$.
\end{myclaim}
Suppose by contradiction that $e(u) < e(v)$.
In particular, $e(u) = 2$ and $e(v) = 3$.
Let $w \in F(v)$.
Note that $w \in S$.
In particular (since $u$ and $v$ are adjacent), $d(u,w) = 2$ and $u \in N(v) \cap I(v,w)$.
Observe that $N(v) \cap C \subseteq N(v) \cap I(v,w)$.
Since $G$ is chordal, $N(v) \cap I(v,w)$ is a clique -- {\it e.g.}, see~\cite{CDK03}.
It implies that $N(v) \cap C \subseteq N(u)$. 
As we further assume that $L(u) \trianglelefteq L(v)$ it follows that $N(u) \cap C = N(v) \cap C$.

Let $y \in N(u) \cap N(w)$.
Since $y \notin N(v)$ and $N(u) \cap C = N(v) \cap C$, $y \notin N[x_1]$.
Furthermore, $\{y\} \cup (N(w) \cap C) \subseteq N(w) \cap I(w,v)$, that implies $N(w) \cap C \subseteq N(y)$ because $G$ is chordal~\cite{CDK03}.
By Lemma~\ref{lem:chordal-inclusion-proj}, $N(u) \cap C$ and $N(y) \cap C$ are comparable.
Therefore, $N(u) \cap C \subseteq N(y) \cap C$.
Note in particular that $\ell(y) > \ell(u)$.
But, since we assume $L(u) \trianglelefteq L(v)$ and $y \notin N(v)$, there must exist a vertex $z \in N(v) \setminus N(u)$ such that $\ell(z) \geq \ell(y)$.

In this situation, we prove that $N(y) \cap C \subseteq N(z)$.
Suppose by contradiction $N(y) \cap C \not\subseteq N(z)$.
Then, $y$ and $z$ are nonadjacent (otherwise, by Lemma~\ref{lem:chordal-inclusion-proj} $N(y) \cap C$ and $N(z) \cap C$ are comparable, and therefore $N(y) \cap C \subseteq N(z)$ because we also have $\ell(z) \geq \ell(y)$).
Let $a \in (N(y) \cap C) \setminus N(z)$ and let $b \in (N(z) \cap C) \setminus N(y)$ (such a vertex $b$ must exist since we have $\ell(z) \geq \ell(y)$).
Again according to Lemma~\ref{lem:chordal-inclusion-proj} we have that $N(v) \cap C$ and $N(z) \cap C$ are comparable, and therefore $N(v) \cap C \subseteq N(z) \cap C$. But then, $N(v) \cap C = N(u) \cap C \subseteq N(y) \cap N(z)$, and therefore, $a,b \notin N(v) \cup N(u)$.
It implies that $(y,a,b,z,v,u)$ induces a cycle, thus contradicting that $G$ is chordal.

Finally, since we have $N(w) \cap C \subseteq N(y) \cap C \subseteq N(z) \cap C$, $z \in N(v) \cap I(v,w)$.
Since $G$ is chordal, $u$ and $z$ must be adjacent~\cite{CDK03}, a contradiction. $\diamond$

\smallskip
In the same way, correctness of the second phase directly follows from the following claim:
\begin{myclaim}
If $u,v \in S^*$ satisfy $u \prec_N v$, then $e(u) \geq e(v)$.
\end{myclaim}
Indeed, suppose by contradiction $e(v) = 3$ whereas $e(u) = 2$.
Let $w \in F(v)$ be arbitrary.
If $uw \in E$, then since we also have $u \prec_N v$ and $w \notin N(v)$, we get $N(w) \cap C \subseteq N(v)$.
In particular, $d(v,w) \leq 2$, that is a contradiction.
Therefore, $d(u,w) = 2$.
Let $y \in N(u) \cap N(w)$.
Since $u \prec_N v$, we have $N(u) \cap C \subseteq N(v)$.
In particular, $y \notin N[x_1]$.
In the same way, since we have $u \prec_N v$ and $y \notin N(v)$ (else, $d(v,w) \leq 2$), we get $N(y) \cap C \subseteq N(v)$.
But then, according to Lemma~\ref{lem:chordal-inclusion-proj}, $N(y) \cap C$ and $N(w) \cap C$ are comparable.
Therefore, $N(v) \cap N(w) \cap C \neq \emptyset$, that is a contradiction. $\diamond$

\smallskip
\noindent
{\bf Complexity.}
Since it is an essential aspect to our algorithm, we start proving that:

\begin{myclaim}\label{lem:l-order}
An $L$-ordering can be computed in linear time.
\end{myclaim}
We maintain an ordered partition ${\cal P}$ of $S$, that is initially reduced to ${\cal P} = (S)$.
Our procedure ensures that at any of its steps, we can order the vertices of each group in ${\cal P}$ so as to obtain an $L$-ordering.
In particular, at the end of the procedure we can output an $L$-ordering simply by listing all the groups of ${\cal P}$ in order (vertices inside a same group can be ordered arbitrarily).
For that, we may order in linear time the vertices of $S$ by non-decreasing $\ell(\cdot)$-values -- {\it e.g.}, by using counting sort.
In doing so, we partition $S$ into disjoint subsets $S_1, S_2, \ldots, S_p$ such that: for every $u,v \in S_i$ we have $\ell(u) = \ell(v) = \ell_i$; and $\ell_1 > \ell_2 > \ldots > \ell_p$.
We actualize the partition as ${\cal P}=(S_p,S_{p-1},\ldots,S_2,S_1)$.
Then, we consider the sets $S_i$ sequentially, from $S_1$ to $S_p$ by increasing index $i$.
Let $N_i$ be the set of all vertices of $S$ with at least one neighbour in $S_i$, that we further partition in $N_{i,1}, N_{i,2}, \ldots, N_{i,|S_i|}$ so that, for any $j \geq 1$, the vertices of $N_{i,j}$ exactly have $j$ neighbours in $S_i$.
Being given the current partition ${\cal P}$, we refine each of its groups $X$ as the ordered sub-sequence $X \setminus N_i, X \cap N_{i,1}, X \cap N_{i,2}, \ldots, X \cap N_{i,|S_i|}$ -- eventually removing any empty group from this sub-sequence.
This can be done in total ${\cal O}(|N_i|)$ time by using classical partition refinement techniques~\cite{HMPV00,PaT87}.
Overall, since all the subsets $S_i$ are pairwise disjoint, the total runtime is linear. $\diamond$

\smallskip
In particular, being already given an $L$-ordering, Phase 1 of the algorithm (reduction from $S$ to $S^*$) clearly runs in linear time.
The difficulty consists in analyzing the runtime of the second phase.
Let $A$ contains all vertices $u \in S^*$ that are processed during this second phase. 
Let $T(m)$ be the runtime of a procedure outputting, for a given vertex $u \in S^*$, all vertices $v \in S^*$ such that $u \prec_N v$. 
The runtime of the second phase is in ${\cal O}(|A|T(m))$ time.

We prove next that $T(m) = {\cal O}(m)$.
In what follows, let $u \in S^*$ be arbitrary.
\begin{itemize}
    \item We assign some counter $\gamma(v)$ to each $v \in V$ -- initially, $\gamma(v) := 0$.
          We stress here that both vertices in $C$ and in $S$ are assigned such a counter.
    \item We consider each $y \in N[u] \cap S$ sequentially. For every $z \in N(y) \cap C$ we increment the counter $\gamma(z)$ of one unit. If furthermore $y \neq u$, then we also consider every $v \in N(y) \cap S^*, \ v \neq u$, sequentially. If $\ell(v) < \ell(y)$, then we increment $\gamma(v)$ of $\ell(y) - \ell(v)$ units.
    \item We discard all vertices $v \in S^* \setminus \{u\}$ such that $\gamma(v) + \sum_{z \in N(v) \cap C} \gamma(z) = \sum_{y \in N[u] \cap S} \ell(y)$.
\end{itemize}
Each step runs in linear time.
Correctness of the procedure follows from the next Claim:

\begin{myclaim}\label{claim:prec-N}
For every $v \in S^*$, $u \prec_N v$ if and only if $\gamma(v) + \sum_{z \in N(v) \cap C} \gamma(z) = \sum_{y \in N[u] \cap S} \ell(y)$.
\end{myclaim}
For short, let us write $\Gamma(v) = \gamma(v) + \sum_{z \in N(v) \cap C} \gamma(z)$.
We first consider a vertex $y \in N(u) \cap N(v) \cap S$ -- if any. 
By Lemma~\ref{lem:chordal-inclusion-proj}, $N(v) \cap C$ and $N(y) \cap C$ are comparable for inclusion.
If $\ell(v) \geq \ell(y)$, then $N(y) \cap C \subseteq N(v)$. 
At the time when we processed vertex $y$, we incremented each counter $\gamma(z), \ z \in N(y) \cap C$ of one unit, which all contribute to $\Gamma(v)$.
Else, $\ell(v) < \ell(y)$, and so $N(v) \cap C \subseteq N(y)$.
In this situation only the vertices of $N(v) \cap C$ contribute to $\Gamma(v)$, however this is compensated by the direct increment of $\gamma(v)$ of exactly $\ell(y) - \ell(v)$ units.
Summarizing, the processing of any vertex $y \in N(u) \cap N(v) \cap S$ always results in increasing $\Gamma(v)$ of exactly $\ell(y)$ units.
Finally, let us consider a vertex $y \in N[u] \cap S$ such that $y \notin N(v)$ -- there is at least one such vertex, namely $u$.  
The processing of $y$ can only result in increasing $\Gamma(v)$ by at most $\ell(y)$ units.
Moreover, this increase is of exactly $\ell(y)$ units if and only if all vertices of $N(y) \cap C$ contribute to $\Gamma(v)$, {\it i.e.}, $N(y) \cap C \subseteq N(v)$. $\diamond$

\smallskip
By Claim~\ref{claim:prec-N}, the runtime of the second phase (and so, of the whole algorithm) is in ${\cal O}(|A|m)$. We prove next that $|A| \leq k$.
For that, it suffices to prove that $A$ is an asteroidal set.
This is done in two steps.
First, we prove that all vertices of $A$ are pairwise incomparable with respect to the binary relation $\prec_N$.
\begin{myclaim}
Let $u,v \in S^*$ be such that $L(u) \trianglelefteq L(v)$.
If $v \prec_N u$, then $u \prec_N v$.
\end{myclaim}
Indeed, we have $N(v) \cap C \subseteq N(u)$. 
Since we assume $L(u) \trianglelefteq L(v)$, it follows that $\ell(u) = \ell(v)$ and therefore $N(v) \cap C = N(u) \cap C$.
Let us now assume that $N(u) \setminus N(v) \neq \emptyset$ (else we are done).
Let $L'(u)$ (resp., $L'(v)$) be the restriction of $L(u)$ (resp., of $L(v)$) to the ordered sub-sequence of labels for $N(u) \setminus N(v)$ (resp., for $N(v) \setminus N(u)$).
Observe that $L(u) \trianglelefteq L(v)$ implies $L'(u) \trianglelefteq L'(v)$.
In particular, $N(v) \setminus N(u) \neq \emptyset$. 
Let $x \in N(u) \setminus N(v)$ be arbitrary and let $y \in N(v) \setminus N(u)$ be maximizing $\ell(y)$.
Since we have $L'(u) \trianglelefteq L'(v)$, $\ell(x) \leq \ell(y)$.
Furthermore, $\ell(y) \leq \ell(u)$ because $N(y) \cap C \subseteq N(u)$.
By Lemma~\ref{lem:chordal-inclusion-proj}, $N(x) \cap C$ and $N(u) \cap C$ are comparable.
As a result, $N(x) \cap C \subseteq N(u) \cap C = N(v) \cap C$. $\diamond$

\smallskip
We end up proving that pairwise incomparable vertices of $S^*$ with respect to $\prec_N$ always form an asteroidal set of $G$.
\begin{myclaim}
Let $A \subseteq S^*$ satisfy that $u \not\prec_N v$ for every $u,v \in A$ distinct.
Then, $A$ is an asteroidal set of $G$.
\end{myclaim}
In order to prove the claim, let $v \in A$ be arbitrary.
Let $C' = C \setminus N(v)$. We stress that $C' \neq \emptyset$ (otherwise, $u \prec_N v$ for any $u \in S^* \setminus \{v\}$).
Hence, for every $u \in A \setminus \{v\}$, it suffices to prove the existence of a path between $u$ and $C'$ in $G \setminus N[v]$.
-- Indeed, doing so, we shall obtain that $A \setminus \{v\}$ is contained in some connected component of $G \setminus N[v]$. --
For that, let $u' \in N[u] \setminus N(v)$ such that $N(u') \cap C \not\subseteq N(v)$ (such a vertex must exist because we assume that $u \not\prec_N v$).
By the very choice of $u'$, this vertex has some neighbour in $C'$, that implies the existence of a path of length at most two between $u$ and $C'$. $\diamond$
\end{proof}

\section{Conclusion}\label{sec:ccl}
{We generalized most known algorithmic results for diameter computation within AT-free graphs to the graphs within the much larger class $Ext_{\alpha}$, for any $\alpha \geq 2$.
The AT-free graphs, but also every DP graph with diameter at least six, belong to $Ext_2$. 
Furthermore, for every $\alpha \geq 3$, every graph of asteroidal number $\alpha$ (and so, every $\alpha$-polygon graph, every $\alpha$-moplex graph and every chordal graph of leafage equal to $\alpha$) belong to $Ext_{\alpha}$.
Evidently, every graph belongs to $Ext_{\alpha}$, for some $\alpha \geq 2$.

We left open whether our results could be extended to the problem of computing exactly all eccentricities. 
The algorithm in~\cite{Duc21c} for the AT-free graphs can also be applied to this more general setting.
Specifically, in every AT-free graph, there exist three vertices $u,v,w$ such that every vertex $x$ is at distance $e(x)$ from a vertex in $N[u] \cup N[v] \cup N[w]$.
An algorithm is proposed in~\cite{Duc21c} in order to prune out each of the three neighbourhoods to a clique while keeping the latter property.
Now, being given a graph from $Ext_{\alpha}$, using our techniques we could also compute a union of ${\cal O}(\alpha^2)$ neighbourhoods such that every vertex $x$ is at maximum distance from some vertex in it.
However, insofar our approach only allows us to extract from each neighbourhood a vertex of maximum eccentricity, that is less powerful than the pruning method in~\cite{Duc21c}.

It could be also interesting to give bounds on the number of extremities (resp., of pairwise nonadjacent extremities) in other graph classes.
Finally, we left as an open problem what the complexity of computing the diameter is within the graphs which can be made AT-free by removing at most $k$ vertices.
For stars and their subdivisions it suffices to remove one vertex, and therefore this class of graphs has unbounded asteroidal number even for $k=1$.
}

\bibliography{biblio}

\appendix

\section{Certification of our algorithms}\label{app:compute-alpha}

Our algorithms in Sec.~\ref{sec:alg} require, being given a prime graph $G$, to be also given some explicit constant $\alpha$ such that $G \in Ext_{\alpha}$.
Next, we explain how this assumption can be removed.

Recall the following procedure common to both Theorems~\ref{thm:approx} and~\ref{thm:exact}.
Starting from some almost central vertex $c$, we repeatedly compute pairwise nonadjacent extremities $x_1,x_2,\dots,x_t$ and shortest $cx_i$-paths $P_i$, for every $1 \leq i \leq t$, until their union $H = \bigcup_{i=1}^t P_i$ becomes a dominating set of $G$. The runtime of this procedure is in ${\cal O}(tm)$, that is in ${\cal O}(\alpha m)$ provided we have $G \in Ext_{\alpha}$.
We then need to consider ${\cal O}(\alpha)$ vertices on each shortest path $P_i$ (specifically, we need to consider $66\alpha-19$ and $42\alpha-11$ such vertices for Theorems~\ref{thm:approx} and~\ref{thm:exact}, respectively). For each vertex considered, either we execute a BFS (Theorem~\ref{thm:approx}) or we apply Lemma~\ref{lem:search-diam} (Theorem~\ref{thm:exact}). Neither case requires to know $\alpha$. Therefore, the value of $\alpha$ is only needed in order to compute the number of vertices to be considered on each shortest path $P_i$.
It turns out this number of vertices actually depends on the hyperbolicity $\delta$ of the graph.
Specifically, for each shortest path $P_i$, we either need to consider $22\delta+3$ (Theorem~\ref{thm:approx}) or $14\delta+3$ vertices (Theorem~\ref{thm:exact}).
In particular, if we can compute some upper bound $\delta^*$ for the hyperbolicity so that $\delta^* = {\cal O}(\alpha)$, then we needn't be given $\alpha$ in the input.

\medskip
Let $s$ be an arbitrary vertex.
For every $u,v \in V$, we write $u \sim^s v$ if and only if the following conditions hold:
\begin{enumerate}
    \item $d(u,s) = d(v,s) = j$, for some $0 \leq j \leq e(u)$;
    \item and there exists a $uv$-path in $G \setminus N^{j-1}[s]$.
\end{enumerate}
Let $\Delta_s(G) = \max\{ d_G(u,v) \mid u \sim^s v \}$.

\begin{lemma}[\cite{AbD16}]\label{lem:cluster-diam-1}
For every vertex $s$ in a graph $G$, the hyperbolicity of $G$ is at most $\Delta_s(G)$.
\end{lemma}

\begin{lemma}[Corollary 2 \& Lemma 4 in~\cite{CDNR+12}]\label{lem:cluster-diam-2}
For every vertex $s$ in a graph $G$, if $G$ can be embedded into a tree with additive stretch $\lambda$, then we have $\frac{\Delta_s(G)}{3}-1 \leq \lambda \leq \Delta_s(G) + 2$.
\end{lemma}

We recall that for a prime graph in $Ext_{\alpha}$, $\lambda \leq 3\alpha-1$~\cite{KKM01}.
Therefore, by Lemmas~\ref{lem:cluster-diam-1} and~\ref{lem:cluster-diam-2}, we are left computing a constant-factor approximation of $\Delta_s(G)$ for some suitable choice of start vertex $s$.
We do so in what follows for $s=c$.

\smallskip
\noindent
Our algorithm goes as follows:
\begin{itemize}
    \item We compute the equivalence classes $C_1,C_2,\ldots,C_p$ of $\sim^c$. 
    \item For every $1 \leq j \leq p$, let $Lab(C_j) = \{ i \mid 1 \leq i \leq t, \ N[P_i] \cap C_j \neq \emptyset \}$.
    \item We output $\Delta^*:= \max_{1 \leq j \leq t}\max\{ 2d(v,P_i)+2 \mid v \in C_j, \ i \in Lab(C_j) \}$.
\end{itemize}

\smallskip
\noindent
{\bf Correctness.}
We prove in what follows that we have $\Delta^*/2 - 2 \leq \Delta_c(G) \leq \Delta^*$.

First, we pick some class $C_j$, for $1 \leq j \leq p$, such that $\max\{ d(v,P_i) \mid v \in C_j, \ i \in Lab(C_j) \}$ is maximized. 
By the maximality of $C_j$, we have $\Delta^* = 2 \times \max\{ d(v,P_i) \mid v \in C_j, \ i \in Lab(C_j) \} + 2$.
Let $v \in C_j$ and $i \in Lab(C_j)$ be maximizing $d(v,P_i)$.
By construction, there exists a vertex $u \in C_j \cap N[P_i]$.
Then, $\Delta_c(G) \geq d(u,v) \geq d(v,P_i) - 1 = \Delta^*/2 -2$.

Second, we pick some class $C_{j'}$, for $1 \leq j' \leq p$, such that $\max\{d(u,v) \mid u,v \in C_{j'}\}$ is maximized.
By the maximality of $C_{j'}$, we have $\Delta_c(G) = \max\{d(u,v) \mid u,v \in C_{j'}\}$.
Let $u,v \in C_{j'}$ be maximizing $d(u,v)$.
We pick an arbitrary label $i$, for some $1 \leq i \leq t$, such that $u \in N[P_i]$ (such a label always exists because $H = \bigcup_{i=1}^t P_i$ is a dominating set).
In particular, we have $i \in Lab(C_{j'})$.
Now, let $x,y \in P_i$ satisfy $d(v,x) = d(v,P_i)$ and $y \in N[u]$.
Note that $|d(c,x) - d(c,v)| \leq d(v,P_i)$ and that $|d(c,u)-d(c,y)| \leq 1$.
As a result,
\begin{flalign*}
\Delta_c(G) &= d(u,v) \leq d(v,P_i) + d(x,y) + 1 \\
&= d(v,P_i)+|d(c,x)-d(c,y)|+1 \\
&\leq d(v,P_i) + |d(c,x)-d(c,v)| + |d(c,v)-d(c,y)| + 1\\
&\leq d(v,P_i) + |d(c,x)-d(c,v)| + |d(c,u)-d(c,y)| + 1\\
&\leq 2d(v,P_i) + 2 \leq \Delta^*.
\end{flalign*}

\smallskip
\noindent
{\bf Complexity.}
The equivalence classes of $\sim^c$ can be computed in linear time~\cite{BCD99}.
In order to compute the label-sets $Lab(C_1),Lab(C_2),\ldots,Lab(C_p)$, we proceed as follows:
\begin{itemize}
    \item We set $Lab(v) := \emptyset$ for every vertex $v$.
    \item We scan each shortest path $P_i$, for $1 \leq i \leq t$, sequentially. For each vertex $v \in V(P_i)$, for each vertex $u \in N[v]$, we add label $i$ to $Lab(u)$.
    \item For every $1 \leq j \leq p$, we set $Lab(C_j) := \bigcup_{v \in C_j}Lab(v)$.
\end{itemize}
The runtime of this above procedure is in ${\cal O}(t m + t n) = {\cal O}(\alpha m)$.

Then, for every fixed $P_i$, $1 \leq i \leq t$, a standard modification of BFS allows us to compute all distances $d(v,P_i), \ v \in V$ in linear time.
In particular, the total runtime in order to compute all distances $d(v,P_i)$, for $v \in V$ and $1 \leq i \leq t$, is in ${\cal O}(tm) = {\cal O}(\alpha m)$.
Now, for every $1 \leq j \leq p$, we can compute $\max\{ d(v,P_i) \mid v \in C_j, \ i \in Lab(C_j) \}$ in ${\cal O}(t|C_j|) = {\cal O}(\alpha |C_j|)$ time.
Recall that $\sum_{j=1}^p |C_j| = n$.
Overall, the total runtime of the algorithm is in ${\cal O}(\alpha m)$.

\end{document}